\documentclass[11pt]{article}
\usepackage{amssymb, latexsym, mathtools, amsthm}
\begingroup
    \makeatletter
    \@for\theoremstyle:=definition,remark,plain\do{%
        \expandafter\g@addto@macro\csname th@\theoremstyle\endcsname{%
            \addtolength\thm@preskip\parskip
            }%
        }
\endgroup
\usepackage{enumerate}
\usepackage{pgf,tikz}
\usepackage{pdfsync}
\usepackage{txfonts}
\usepackage{datetime}
\usepackage[T1]{fontenc}
\usepackage{booktabs}
 \usepackage{graphicx}
\definecolor{dnrbl}{rgb}{0,0,0.3}
\definecolor{dnrgr}{rgb}{0,0.3,0}
\definecolor{dnrre}{rgb}{0.5,0,0}
\usepackage[colorlinks=true, citecolor=dnrgr, linkcolor=dnrre, urlcolor=dnrre] {hyperref}
\usepackage{xcolor}
\usepackage{parskip}
\usepackage{tabularx, colortbl}
\usepackage{geometry}
\theoremstyle{plain}
\newtheorem{thm}{Theorem}[section]
\newtheorem{prop}[thm]{Proposition}

\newtheorem{lem}[thm]{Lemma}

\theoremstyle{definition}

\newtheorem{defi}[thm]{Definition}
\makeatletter

\let\c@table\c@figure
\makeatother

\newcommand{\twolo}{2^{<\omega}}

\newcommand{\Nat}{\mathbb{N}}

\newcommand{\restr}{\upharpoonright}  
 
\newcommand{\un}{\uparrow} 
\newcommand{\de}{\downarrow} 
\newcommand{\intg}{\mathtt{Int}_g}
\DeclarePairedDelimiter{\ceil}{\lceil}{\rceil}
\DeclarePairedDelimiter{\floor}{\lfloor}{\rfloor}

\newcommand{\bigo}[1]{\mathop{\bf O}\/\left({#1}\right)}

\newcommand{\eg}{e.g.\ }
\newcommand{\ie}{i.e.\ }

\newcommand{\twome}{2^{\omega}} 
\newcommand{\twomel}{2^{<\omega}}

\renewenvironment{abstract}
 { \normalsize
  \list{}{
    \setlength{\leftmargin}{.0cm}%
    \setlength{\rightmargin}{\leftmargin}%
    }%
  \item {\bf \abstractname.} \relax}
 {\endlist}
\usepackage{titlesec}
\usepackage[round]{natbib}   

 \titleformat{\paragraph}
{\normalfont\normalsize\em\bfseries}{\theparagraph}{1em}{$\blacktriangleright$\hspace{0.2cm}}
\titlespacing*{\paragraph}{0pt}{3.25ex plus 1ex minus .2ex}{0.5ex plus .2ex}

\title{Granularity of wagers in games and the possibility of saving
\thanks{Barmpalias was supported by the 
1000 Talents Program for Young Scholars from the Chinese Government No.~D1101130, 
NSFC grants 11750110425 and 11971501.
Fang Nan was supported by the China Scholarship Council (Ministry of Education). 
Support by the Jiangsu Provincial Advantage Fund is acknowledged for a visit to the 
University of Nanjing during April 2019. We wish to thank the referees for useful feedback which improved
the presentation of our article.}}

\author{George Barmpalias  \and Nan Fang}
\date{\today\  at\ \currenttime}
\begin{document}
\maketitle
\begin{abstract}
In a casino where arbitrarily small bets are admissible, 
any betting strategy $M$
can be modified into a saving strategy that, not only is 
successful on each casino sequence where $M$ is
(thus accumulating unbounded wealth inside the casino) but
also saves an unbounded capital, by
permanently and gradually withdrawing it from the game.  
Teutsch showed that this is no longer the case
when a fixed minimum wager is imposed by the casino, thus exemplifying a
 {\em savings paradox} where a player can win unbounded wealth
inside the casino, but upon withdrawing a sufficiently large amount out 
of the game, he is forced into bankruptcy. 
We study the potential for 
saving under a shrinking minimum wager rule (granularity)
and its dependence on the rate of decrease (inflation)
as well as  {\em timid versus bold play.}
\end{abstract}
\vspace*{\fill}
\noindent{\bf George Barmpalias}\\[0.5em]
\noindent
State Key Lab of Computer Science, 
Institute of Software, Chinese Academy of Sciences, Beijing, China.\\[0.2em] 
\textit{E-mail:} \texttt{\textcolor{dnrgr}{barmpalias@gmail.com}}.
\textit{Web:} \texttt{\textcolor{dnrre}{http://barmpalias.net}}\par
\addvspace{\medskipamount}\medskip\medskip

\noindent{\bf Nan Fang}\\[0.2em]
\noindent Institut f\"{u}r Informatik, Ruprecht-Karls-Universit\"{a}t Heidelberg, Germany.\\[0.1em]
\textit{E-mail:} \texttt{\textcolor{dnrgr}{nan.fang@informatik.uni-heidelberg.de}.}
\textit{Web:} \texttt{\textcolor{dnrre}{http://fangnan.org}} \par
\addvspace{\medskipamount}\medskip\medskip

\vfill \thispagestyle{empty}
\clearpage

\setcounter{tocdepth}{2}
\tableofcontents

\section{Introduction}\label{1JyMSDjB2p}
In a casino where a fixed minimum wager is imposed on the bets, a 
player may be forced to quit the game due to insufficient capital required for joining the next round,
while still in possession of a non-zero sum.
This basic fact was exploited by  
\citet{cie/BienvenuST10} in an investigation of the strength of effective 
betting strategies that have restrictions on the
admissible wagers. This work motivated further studies on the power of 
restricted wager strategies, beyond the original
algorithmic framework, in the case
where the restriction is fixed throughout the game.
Given a set of reals $X$, an $X$-valued strategy is one that is restricted on wagers in $X$.
Given two finite sets $A,B$ of rationals, by
 \cite{TeutChalcraft}, $A$-valued strategies 
can successfully replace any $B$-valued strategy, if and only if 
there exists $r\geq 0$ such that $B\subseteq r\cdot A$
(where $r\cdot A$ denotes the multiples of the elements of $A$ with $r$).
In particular, subject to the given condition, given 
any strategy restricted to bets in $B$, we can produce a strategy 
that only bets values in $A$ and succeeds 
(producing  unbounded wealth) 
on any casino outcome sequence where the $B$-restricted strategy succeeds.
This characterization was extended to infinite sets, with some 
additional conditions, in \cite{Peretzagainst}.
Remarkably,
\citet{Teu14agCWGnp} (also see \cite[Theorem 14]{Peretzwager} for the corrected argument) 
constructed a casino which allows integer-wager strategies to succeed,
producing unbounded
wealth inside the casino, but any player who attempts to save 
an unbounded amount by removing it from the casino, is
forced to bankruptcy.
Motivated by these developments we
consider  {\em granular strategies} which are restricted to certain {\em discrete}, but not necessarily
integer, wagers and
\begin{equation}\label{VnOfYbkXhT}
\parbox{11cm}{study the potential for 
saving in betting strategies,
and its dependence on the
{\em granularity} of the wagers, as well as  
{\em timid versus bold play.}}
\end{equation}
%
By a (potentially biased) {\em casino} we mean a set of infinite binary sequences which represent
the sequences of possible binary outcomes
in a repeated betting game, along with possible restrictions on the admissible wagers at each stage. 
A betting strategy is a function that, given an initial capital, determines the wager and the
favorable outcome, given any position (represented by the binary string of the previous outcomes) 
in the game. A strategy is successful along a casino sequence if
along the game its capital is unbounded.
A {\em saving strategy} is a strategy along with a non-decreasing {\em savings function}  which indicates
the part of the capital at each position of the game which 
is saved, hence permanently removed from the active capital
of the strategy that can be used for betting.
A {\em saving strategy is successful} along a casino 
sequence if its savings function is unbounded.
A strategy (or saving strategy) is successful in a casino 
if its wagers meet the restriction of the casino and
it is successful in all outcome sequences of the casino.\footnote{Such casinos with 
restricted possible 
outcome-sequences were termed `probability-free' in
\cite{TeutChalcraft}.}
Given a sequence of reals $g=(g_s)$
a strategy is $g$-granular if it wagers an integer multiple of the {\em granule} $g_s$, which 
may be interpreted as the {\em value or purchasing power of one currency unit}
at stage $s$. The decrease of $g$ during the game  may be interpreted as the result of inflation.
Assuming that $g$ is non-increasing,
the {\em granularity} in \eqref{VnOfYbkXhT} refers to the
rate of decrease of $g$, indicating the  {\em inflation rate}.
We only distinguish between {\em fine granularity}: $\sum_s g_s<\infty$, and
{\em coarse granularity}: $\sum_s g_s=\infty$.
Alternatively, $g_s$ could be viewed as a minimum wager.

{\bf Our contribution.}
We show that under a rapidly decreasing
granularity, winning in an infinite play 
is equivalent to saving  unbounded capital during the successful game:
\begin{equation}\label{8GjGltExej}
\parbox{13.7cm}{{\em Possibility for saving:} under fine granularity, 
any betting strategy $M$ can be replaced with a saving strategy that
is successful on every  outcome stream where $M$ is successful.}
\end{equation}
Under the inflation interpretation of granularity discussed above, 
\eqref{8GjGltExej} says that saving is equivalent to
winning inside the casino, under a sufficiently high inflation environment.
As a converse of \eqref{8GjGltExej}, we show:
\begin{equation}\label{O2hEsN8zy}
\parbox{13.7cm}{{\em Impossibility of uniform saving:}  Under coarse granularity, 
there exists a betting strategy $M$ such that any saving strategy $N$
fails to save on some $X$ where $M$ is successful.}
\end{equation}
A different converse of \eqref{8GjGltExej}  concerns
{\em timid} strategies where the wagers are $\bigo{g}$ small.
\begin{equation}\label{YIe3MDo7Ue}
\parbox{13.7cm}{{\em Impossibility for timid strategies:}  Under coarse granularity, 
there exists $X$ and a timid betting strategy 
succeeding  on  $X$, such that no timid saving strategy succeeds on $X$.}
\end{equation}
It is customary to restrict the choice of strategies 
amongst a countable collection, typically representing the
feasible or implementable strategies. 
From an algorithmic 
perspective, as did many of the authors cited above, 
we may consider strategies that
are computable, in the sense that they can be simulated by a 
Turing machine. 
In this fashion, we adopt the semantics
of  \cite{Peretzagainst}, which also applies to most of the previous references:
for statements $\forall T\ \exists M$  about strategies $T,M$, 
strategy $M$ is computable in $T$, while in 
$\exists M \ \forall T$ strategy $M$ is computable and $\forall T$ ranges over a countable collection of
strategies $T$ that has been fixed in advance.
Hence the saving strategy of \eqref{8GjGltExej} is computable in the given $M$. In
\eqref{YIe3MDo7Ue} the betting strategy claimed is computable while
the universal `no timid saving strategy succeeds' refers to any fixed-in-advance countable
collection of saving strategies.
If we allow `bold play', \ie strategies that can bet
arbitrarily large wagers,  saving under coarse granularity
can be salvaged, in a weak sense.
\begin{equation}\label{B8qdW8Mt1}
\parbox{12cm}{{\em Saving with large wagers:} 
Under coarse granularity, given any timid
betting strategy $M$ there exists a countable family $(T_i)$
of saving strategies, such that for each $X$ where $M$ is successful, there exists
$i$ such that $T_i$ is successful on $X$.}
\end{equation}
%
The {\em integer-valued strategies} (where wagers are required to be integers)
have been studied in many articles \cite{cie/BienvenuST10, JCSSbarmp15, Herbert15},
and are
a special case of $g$-granular strategies.
Granular strategies have also played a crucial role in the analysis of restricted oracle computations
from algorithmically random sources \cite{bakugaopt}. Roughly speaking, oracle-computations
with oracle-use $n\mapsto n+g(n)$ correspond to $g$-granular strategies: in this sense, understanding
how the granularity of a betting strategy restricts its power can be used  to study the impact of
oracle-access restrictions  a computation.
Regarding our choice of monotone granularity, we note that some existing arguments 
regarding restricted wager strategies 
such as \cite[Theorem 14]{Peretzwager} actually use wagers of decreasing granularity.

{\bf Mathematics of casino games.}
Beyond the historical or illustrative examples of gambling problems that one finds in
probability texts, there are two systematic studies that establish a formal
mathematical framework for the stochastic analysis of  gambling. 
The first one is
based on measure theory and integration, is developed in the monograph
\cite{howgamble} and a considerable number of related articles such as
\cite{meagambart}, and is arguably rooted in the classic \cite{halmos1939}. The simple binary
outcome bets that we study in the present article are known as {\em red-and-black games}
in this line of research \cite[Chapter 5]{howgamble}, and conceptually relevant themes 
to the present work include
{\em permitted sets of bets} \cite[\S 7.3]{howgamble},
{\em timid versus bold play} in \cite{timidlargeart,timidloptiart,genboldart},
and income-tax or stake-grabbing casinos  \cite[\S 9.2--\S 9.3]{howgamble}
which are related to the saving concepts of the recent 
\cite{Teu14agCWGnp,Peretzwager} discussed above  and our study.

The second approach is more constructive and is based on an alternative  
game-theoretic foundation of probability developed in \cite{shafervovkprobfin}
and is rooted on von Mises' incomplete frequentist 
foundation for probability \cite{MR21400mises,misesbook,misesMScILLC},
suitably generalized by the notion of martingales. 
This strategy-based approach is closer to our analysis
and the recent  \cite{Teu14agCWGnp,Peretzwager}, as well as the 
foundations of algorithmic randomness \cite{Loveland66,Schnorr:71} and \cite[\S 6.3]{rodenisbook}.
Apart from these similarities, we have not found stronger connections between these two
classic approaches and the recent line of research that the present work belongs to, although
a comprehensive comparative analysis would be interesting as we discuss in \S\ref{XXYGMo8eSb}.

Types of betting strategies include 
{\em fixed-wager}, the {\em martingale} and its various modifications, as well as 
{\em proportional betting} in which wagers are set as a fixed proportion of the existing capital.
In the latter case, the choice of the proportion is designed to maximize the expectation of a random variable,
such as the capital, or most often the binary logarithm of the capital, which is known as the 
{\em Kelly criterion}, 
discovered by \cite{kellyorig} as an interpretation of Shannon's concept of information rate 
(it maximizes the expected growth of the capital, as opposed to the capital itself).\footnote{This is
one of the most well-known, widely applied
betting and investment strategies, see \cite{kellylogbook}.}
A survey of betting strategies shows that in many situations, 
{\em the amount you bet is actually more important than what you bet on}.
In addition, the success of 
staking methods often depends on the presence of wager bounds  (minimum or maximum bets)
as well as the initial capital (bankroll) of the player. 
In this context, the present work sheds light on the effects of
granularity in its different interpretations (minimum bets or inflation) in
the success of betting strategies.

{\bf Outline of our presentation.}
Toward a formal expression and proof of our main results
in \S\ref{Bmb4ouelqF}, we give the standard definitions and gambling interpretations
in full accordance with \cite{TeutChalcraft, Teu14agCWGnp,Peretzwager,Peretzagainst},
as well as our notion of granularity and some related basic facts.
In \S\ref{MEchtencE} we give the formal statement of our results as well as the proofs,
except for the proof of the formal analogue of \eqref{B8qdW8Mt1} 
which is given in \S\ref{iPgHQCWM76}.
In particular, \eqref{8GjGltExej} corresponds to
Theorem \ref{KOywTgVWA} while \eqref{O2hEsN8zy} and \eqref{YIe3MDo7Ue}
are included in Theorem \ref{esRYf6OWn5};
our last result \eqref{B8qdW8Mt1} is formalized in Theorem \ref{lL3i4uiQDZ}.
We conclude our presentation in \S\ref{XXYGMo8eSb} with a discussion of our results in the
light of algorithmic randomness, an open question and suggestions for further research on this topic.

\section{Martingales, granularity and savings}\label{Bmb4ouelqF}
Betting strategies are formalized by 
martingales, 
expressing the capital after each betting stage and each casino outcome.
Formally, a martingale is a function $M:\twomel\to\mathbb{R}^+$ with the property that
$2M(\sigma)=M(\sigma\ast 0)+M(\sigma\ast 1)$
for all $\sigma$. 
These deterministic (as opposed to probabilistic) martingales 
provide a formalization of betting 
strategies on an infinite coin-tossing game:
at stage $|\sigma|$ our capital is $M(\sigma)$ and our {\em wager} for the next bet is
$\big(M(\sigma\ast 1)-M(\sigma\ast 0)\big)/2$\footnote{Using the martingale property of $M$ the wager
can also be written as $M(\sigma\ast 1)-M(\sigma)$; for clarity, 
we opt for he first formula as, unlike the latter one, it expresses the
wager in the more general case when $M$ is a supermartingale (see below).}
which can be used in the multiplicative expression of martingales:
\begin{equation}\label{he89h1XJ3Q}
M(\sigma)=M(\lambda)\cdot
\prod_{i<|\sigma|} \left(1+(-1)^{\sigma(i)+1}\cdot \frac{w_M(\sigma\restr_{i})}{M(\sigma\restr_i)}\right)
\hspace{0.5cm}\textrm{where}\hspace{0.3cm}
w_M(\sigma)= \big(M(\sigma\ast 1)-M(\sigma\ast 0)\big)/2
\end{equation}
and $\sigma(i)$ is the value of the bit of $\sigma$ at position $i$ (where the first position is position 0).
If $w_M(\sigma)>0$ then at position $\sigma$ we bet on outcome 1, capital  $|w_M(\sigma)|$; otherwise
we bet the same capital on outcome 0. Hence
$M(\sigma\ast j)=M(\sigma)+(-1)^{\sigma(j)+1}\cdot w_M(\sigma)$ reflects the updated 
capital with respect to either outcome $j\in\{0,1\}$.
Saving strategies are formalized by supermartingales, which are functions
$M:\twomel\to\mathbb{R}^+$ such that
$2M(\sigma)\geq M(\sigma\ast 0)+M(\sigma\ast 1)$
for all $\sigma$. 
Supermartingales can be thought of as strategies (\ie martingales) with the difference that after each bet
there is a certain loss of liquid capital, \ie capital that can be used for betting. 
The {\em marginal savings} of a supermartingale $M$ at $\sigma$ is defined as
$M^{\ast}(\sigma)=M^{\ast}(\sigma)-\big(M(\sigma\ast 0)+M(\sigma\ast 1)\big)/2$
and is the amount that is lost from position $\sigma$ to the next bet, \ie the amount by which $M$ fails
to satisfy the martingale inequality  at $\sigma$.
Given a supermartingale $M$ define {\em the cover of $M$} to be the unique martingale 
whose initial capital and wagers are the same as those of $M$.
\[
\textrm{{\bf Savings of $M$:} \hspace{0.2cm}
$S_M(\sigma)=\widehat{M}(\sigma)-M(\sigma)$, 
\hspace{0.2cm}where $\widehat{M}$ denotes the cover of $M$.}
\]
Clearly $S_M(\sigma)$ is simply the sum of the marginal savings of $M$ on the initial segments of $\sigma$.
The wager of a supermartingale $M$ is also given by \eqref{he89h1XJ3Q}.

{\em Martingales as oracles.} In the following sections we often refer to oracle-computations in which on e of the oracles is a martingale
$M$. For such statements, note that
$M$ is a  function from $\twomel$ to the non-negative reals, so it can be
represented by a real which encodes a fast (\eg with modulus of convergence $n\mapsto 1/n$) rational approximation to $M(\sigma)$ for each $\sigma$.
It is such a representation of $M$ that is used as an oracle in computations from $M$, in the standard sense of
relative Turing computation. Alternatively, the reader may replace each martingale $M$ with a rational-valued martingale $M'$
which has the same asymptotic properties as $M$, and use $M'$ as an oracle (since $M'$ has a more straightforward representation
as a real). The fact that this replacement can be made without loss of generality, is due to a folklore fact 
about effective martingales, see 
\cite[Proposition 7.1.2]{rodenisbook} or \cite{Schnorr:71}.

\subsection{Strategy success, wager scaling and the savings trick}\label{W2nKOLMi5}
We say that a martingale $M$ (as a betting strategy) is {\em successful} 
along $X$ if $\limsup_n M(X\restr_n)=\infty$. We also say
that $M$ {\em successfully saves} (or the associated saving strategy is successful) 
if $S_M(X\restr_n)\to\infty$ as $n\to\infty$.
A folklore and useful fact for the case when the wagers are not required to be discrete is that
successful betting is equivalent to successful saving:
\begin{equation}\label{DxAECJj5D3}
\textrm{\bf Savings trick:}\hspace{0.3cm}\left(\ \ \parbox{10cm}{Each 
supermartingale $M$ computes
a supermartingale $N$ such that
$\lim_n S_N(X\restr_n)=\infty$ for each 
$X$ such that $\limsup_n M(X\restr_n)=\infty$.}\ \ \right)
\end{equation}
The idea behind the saving strategy $N$ 
in \eqref{DxAECJj5D3} is scaling the wagers, and is relevant to
the later sections of this article. Without loss of generality we may
assume that $M(\lambda)>1$.
At the beginning, $N$ bets identically to $M$, until some position of the game is reached where
$M(\sigma)$ is more than the double of the initial capital $M(\lambda)$. At such a position  $\sigma_1$ strategy
$N$ saves 1 (making the difference between the $N$ and $M$ capital equal to 1), and proceeds with the 
subsequent bets proportionally adjusted, where the proportion is $(M(\sigma_1)-1)/M(\sigma_1)$. 
At the next position $\sigma_2$ where $M$ doubles with respect to the previous marked value $M(\sigma_1)$, 
we repeat the same action, letting $N$ save another 1,
and adjusting the subsequent bets proportionally with respect to the ratio
$N(\sigma_2)/M(\sigma_2)$ and so on. 

By the proportionality of bets and 
the multiplicative form \eqref{he89h1XJ3Q} of strategies,
between positions $\sigma_1$ and $\sigma_2$ the ratio $N(\sigma)/M(\sigma)$ remains equal to
$(M(\sigma_1)-1)/M(\sigma_1)$. In particular, at position $\sigma_2$ where $M$ doubles its capital compared to
position $\sigma_1$, the same happens to $N$, compared to $N(\sigma_1)$.
Hence, given that
$N(\sigma_1)=M(\sigma_1)-1>1$, 
 inductively we have
$N(\sigma_n)\geq 2N(\sigma_{n-1})-1> 1$ and $S_N(\sigma_n)=n$ for each $n>1$ where $\sigma_n$ is defined.
Then it is clear that along any $X$ where $M$ is successful, 
the sequence $(\sigma_i)$ of initial segments of $X$ is totally defined,
hence showing the success of the saving strategy $N$ along $X$.
The savings trick implies that the standard success condition 
$\limsup_n M(X\restr_n)=\infty$ for a supermartingale $M$
is essentially equivalent to $\lim_n M(X\restr_n)=\infty$ in the sense
that $M$ computes a
supermartingale $T$ such that 
$\lim_n T(X\restr_n)=\infty$ for each $X$ where $M$ is successful.
%
%
%
\subsection{Discretizing the strategies and the effect on success and savings}\label{sESAlKb5h3}
Intuitively speaking, the `granularity' of a function $f:\twomel\to\mathbb{Q}$ measures how far the values of $f$
are from being integers. For example we may say that the granularity of $f$ is 
the function $g:\Nat\to\Nat$ such that $g(|\sigma|)$ is the minimum
non-negative integer such that $f(\sigma)$ is an integer multiple of $2^{-g(|\sigma|)}$.
Applying this notion to the wagers, we can model a stage-dependent minimum-bet policy
in the casino.
\begin{defi}[Granular martingales]
Given a non-decreasing $g:\Nat\to\Nat$, we say that a (super)martingale $M$ is $g$-granular
if for every string $\sigma$ the wager $w_M(\sigma)$ is an integer multiple of $2^{-g(|\sigma|+1)}$.
\end{defi}
A function $f:\twomel\to\mathbb{Q}^{+}$ is $g$-granular if for each
string $\sigma$ the value $f(\sigma)$ is an integer multiple of $2^{-g(|\sigma|)}$.
One may also consider to apply the notion of granularity to the capital function $\sigma\mapsto M(\sigma)$
of a strategy instead of its wagers, thus obtaining a stronger notion. 
However, as we observe below, such a distinction is not consequential in the present work.
By the above definitions of granularity it follows that
\begin{equation*}
\parbox{13cm}{given a non-decreasing $g:\Nat\to\Nat$ and a $g$-granular martingale $M$,
the function $\sigma\mapsto M(\sigma)$ is $g$-granular if and only if $M(\lambda)$ is
an integer multiple of $2^{-g(0)}$.}
\end{equation*}
Hence given $g$-granular martingale $M$ there exists a martingale $N$ which is computable
from $M,g$, such that 
the function $\sigma\mapsto N(\sigma)$ is $g$-granular and $|M(\sigma)-N(\sigma)|=\bigo{1}$.
More generally, we show that any 
$g$-granular (super)martingale $M$ can be easily transformed into a (super)martingale
which differs by at most a constant from 
$M$, it is $g$-granular as a function, and its savings function takes integer values.
\begin{lem}\label{nXAZQtLDNS}
Given non-decreasing $g:\Nat\to\Nat$ and a $g$-granular supermartingale $M$,
there exists a supermartingale $N$ such that $|M(\sigma)-N(\sigma)|=\bigo{1}$, 
the function $\sigma\mapsto N(\sigma)$ is $g$-granular
and $S_N(\sigma)\in\Nat$ for each $\sigma$. Moreover 
$N$ is computable from $M,g$, and if $M$ is a martingale then $N$
is also a martingale.
\end{lem}
\begin{proof}
Let $T$ be the unique martingale which has the same wagers
as $M$ and $T(\lambda)=\ceil{M(\lambda)}$. Then clearly the function $\sigma\mapsto T(\sigma)$
is $g$-granular, $\widehat{M}(\sigma)\leq T(\sigma)$ and 
$|\widehat{M}(\sigma)-T(\sigma)|=\bigo{1}$.
Define
$N(\sigma)=T(\sigma)-\floor{S_M(\sigma)}$ and note that since $|\widehat{M}(\sigma)-T(\sigma)|=\bigo{1}$ we have
$|N(\sigma)-M(\sigma)|=\bigo{1}$. Since $\widehat{M}(\sigma)\leq T(\sigma)$
we also have $0\leq M(\sigma)\leq N(\sigma)$. By the properties of $T$, the function 
$\sigma\mapsto N(\sigma)$ is $g$-granular. Finally, note that $N$ is computable from $M$ and $g$,
and in the case when $M$ is a martingale we have $S_M(\sigma)=0$ for all $\sigma$, so $N=T$ and $N$ is a martingale.
\end{proof}
Granularity is in conflict with scaling operations on the wagers, so
the saving method of \S\ref{W2nKOLMi5} breaks down in the case of granular strategies.
However  the following property can be salvaged, albeit non-uniformly.
\begin{prop}[Success notions for granular strategies]\label{BrReKBzoKL}
Suppose that $g\colon \Nat\mapsto \Nat$
is nondecreasing and 
$M$ is a $g$-granular supermartingale which is successful on some sequence $X$.
Then there exists a
$g$-granular supermartingale $T$ which is computable from $M$ and $\lim_n T(X\restr_n)=\infty$.
\end{prop}
\begin{proof}
In the case where $\lim_n M(X\restr_n)=\infty$ we can simply let $T:=M$. Otherwise
let $q$ be a positive rational greater than $\liminf_n M(X\restr_n)$ and let $T$ have initial capital $q$.
Let $m_{-1} = 0$.
Then let $T$ produce part of the bets of $M$ along an arbitrary sequence $Y$ as follows:
wait until some $n_0 > m_{-1} $ such that $M(Y\restr_n) < q$ (if such number does not exist, let $n_0=\infty$), and then let $m_0$ be the least $m>n_0$ such that
$M(Y\restr_m)>q+1$ (if such number does not exist, let $m_0=\infty$). 
In the interval $[m_{-1},n_0)$ the strategy $T$ does not place any bets, 
while 
in $[n_0,m_0)$ it places the same bets (\ie the same wagers) that $M$ does, along $Y$.
Hence $T(Y\restr_{n_0}) = T(\lambda) = q > M(Y\restr_{n_0})$
and 
$T(Y\restr_n) = M(Y\restr_n) - M(Y\restr_{n_0}) + T(Y\restr_{n_0}) > M(Y\restr_n) \ge 0$,
for each $n\in [n_0,m_0]$.
Moreover, in the case that $m_0<\infty$, $M(Y\restr_{m_0})-M(Y\restr_{n_0})>1$, so
$T(Y\restr_{m_0}) > T(Y\restr_{n_0})+1 > q+1$. 
This process repeats in the same way, defining the intervals $[m_{i-1},n_i)$ where $T$ does not bet,
and the adjacent intervals $[n_i, m_i)$ where $T$ copies the bets of $M$. If for some $i$ we have
$n_i=\infty$ then after position $Y\restr_{m_{i-1}}$ along $Y$ the value of strategy $M$ never drop below $q$ and the value of $T$ never change.
Or if for some $i$ we have
$m_i=\infty$ then after position $Y\restr_{n_i}$ strategy $T$ cofinaly copies the bets of $M$
along $Y$.
The argument that we used above to show that $T$ is non-negative, inductively shows that 
for each $i\geq 0$ such that $m_i<\infty$ and each $n>n_i$ we have 
$T(Y\restr_{n})>i$. Moreover clearly $T$ is $g$-granular and computable from $M$.
Finally, in the case where $\liminf_n M(Y\restr_n) < q$ and $\limsup_n M(Y\restr_n)=\infty$, the endpoints $n_i,m_i$ are defined for all $i\in\Nat$,
which means that $\lim_n T(Y\restr_{n})=\infty$.
\end{proof}
In the rest of this article $g$ will always denote a function from $\Nat$ to $\Nat$. 
From the proof of Proposition \ref{BrReKBzoKL} we may extract the following useful fact.
\begin{equation}\label{ok15FgmIQl}
\parbox{13cm}{if a computable  $g$-granular strategy succeeds on $X$ but no
such saving strategy does, any successful
computable $g$-granular strategy $M$ on $X$ 
has $\lim_n M(X\restr_n)=\infty$.
}
\end{equation}
Indeed, in the second case of the proof of Proposition \ref{BrReKBzoKL}
where $\liminf_n M(X\restr_n)<\infty$, we essentially make $T$ a saving strategy:
$T$ can be easily modified into a supermartingale $N$ such that $\lim_n S_N(X\restr_n)=\infty$.
Since $N$ only depends on $M$ and a rational upper bound of $\liminf_n M(X\restr_n)$, 
we may conclude that: from any nondecreasing $g$ and $g$-granular supermartingale $M$
we can compute $g$-granular supermartingales $N_q, q\in \mathbb{Q}^+$
(where $\mathbb{Q}^+$ denotes the set of positive rational numbers)
such that for each $X$ where $M$ is successful and $\liminf_n M(X\restr_n)<q$
we have
$\lim_n S_{N_q}(X\restr_n)=\infty$. This is a formal expression of \eqref{ok15FgmIQl}.

\section{The effect of fine or coarse granularity on saving}\label{MEchtencE}
Given nondecreasing $g$, by {\em fine} or {\em coarse} granularity we mean
that $\sum_{i} 2^{-g(i)}$ is finite or infinite, respectively.
We first show \eqref{8GjGltExej} of \S\ref{1JyMSDjB2p}, which informally asserts that
under fine granularity, betting strategies  can be modified so that they also save, whenever they win.
We stress that this modification requires not only
knowledge of $g,M$, but also an upper bound for $\sum_{n} 2^{-g(n)}$, as the following
formal expression indicates.
\begin{thm}[Savings under fine granularity]\label{KOywTgVWA}
Given $G\in\Nat$, nondecreasing $g$ such that  $\sum_{n} 2^{-g(n)}<G$ 
and any  supermartingale $M$,
there exists a $g$-granular supermartingale 
$N$, computable from $G,g,M$, such that $\lim_n S_N(X\restr_n)=\infty$ for each $X$ such that 
$\limsup_n M(X\restr_n)=\infty$.\footnote{With respect to the use of $M$ as an oracle, recall the remarks of \S\ref{Bmb4ouelqF} about
using real-valued martingales as oracles.}
\end{thm}
A  betting strategy $M$ is {\em $g$-timid} if it is $g$-granular and its wager function is $\bigo{2^{-g(n)}}$.
A family $(N_i)$ of saving strategies is a
{\em saving cover} of $M$ if for each $X$ where $M$ succeeds, there exists $i$ such that
$N_i$ successfully saves on $X$. A saving cover $(N_i)$ is {\em bounded} 
if it has finite total initial capital, \ie
$\sum_i N_i(\lambda)<\infty$.\footnote{Note that a bounded family of saving strategies need not be
co-finally trivial (\eg all sufficiently large bets are zero). The reason is that
their initial capitals may all be non-zero, and each strategy could start betting non-zero wagers after
some point when the granularity allows the use of a fraction of their capital. }

We formalize the impossibilities \eqref{O2hEsN8zy} and \eqref{YIe3MDo7Ue} 
 of \S\ref{1JyMSDjB2p}
under coarse granularity as follows.
%
\begin{thm}[Impossibility of bounded or timid countable cover]\label{esRYf6OWn5}
Given $g$ with $\sum_i 2^{-g(i)}=\infty$ there exists $g$-timid $M$, computable from $g$, 
such that for any countable
family $(T_i)$  with one of the following properties:
\begin{enumerate}[\hspace{0.5cm}(a)]
\item  $(T_i)$ is a bounded family of $g$-granular saving strategies
\item  each $T_i$ is a $g$-timid saving strategy
\end{enumerate}
there exists $X$ such that
 $\limsup_n M(X\restr_n)=\infty$ and $\limsup_n S_{T_i} (X\restr_n)<\infty$ for each $i$.
\end{thm}
%
Since the sum of a bounded family of $g$-granular saving strategies is
a $g$-granular saving strategy, clause (a) of 
Theorem \ref{esRYf6OWn5} is a consequence of:
\begin{lem}
\label{fYSGaBvd6u}
If $g$ is nondecreasing and $\sum_{n} 2^{-g(n)}=\infty$,
there exists a $g$-timid martingale $M\leq_T g$ 
such that for every  $g$-granular supermartingale $T$
there exists $X$ with $\lim_n M(X\restr_n)=\infty$ and $\lim_n S_{T}(X\restr_n)<\infty$.
\end{lem}
The proofs of Theorem \ref{esRYf6OWn5} and Lemma \ref{fYSGaBvd6u} 
make heavy use of {\em divisions} that are not Euclidean in the strict sense as 
the numbers involved may not be integers. 
\begin{defi}[Division]\label{zY3U5A7rVM}
Given reals  $t,m$
such that $t\geq 0$, $m>0$ define the {\em quotient} $q$ of the division of $t$ by $m$ as the largest integer
such that $q\cdot m\leq t$, and define the {\em remainder} $r:=t- q\cdot m$, so that $0\leq r<m$. 
\end{defi}
The bulk of the present section is devoted to the proof of these results.
The final \eqref{B8qdW8Mt1} of \S\ref{1JyMSDjB2p} is formalized by:

\begin{thm}[Countable savings cover]\label{lL3i4uiQDZ}
If $g$ is nondecreasing, unbounded and $M$ is $g$-timid betting strategy,
there exists a countable family
$T_i, i\in\Nat$ of saving strategies such that:
\begin{enumerate}[\hspace{0.5cm}(a)]
\item $(T_i)$ is computable from $M$ with wagers  integer multiples of the corresponding wagers of $M$;
\item for each $X$ where $\limsup_n M(X\restr_n)=\infty$ there exists $i\in\Nat$ such that 
$\lim_n S_{T_i} (X\restr_n)=\infty$.
\end{enumerate}
Hence for any $X$ where $M$ is successful, at least one of the $T_i$ saves successfully along $X$.
\end{thm}
The proof of Theorem \ref{lL3i4uiQDZ} is more involved and 
is given in \S\ref{iPgHQCWM76}. It is instructive to exemplify the connections between our results
and the saving paradox of \citet{Teu14agCWGnp}  more formally than we did in \S\ref{1JyMSDjB2p}.

{\em Comparing \citet{Teu14agCWGnp} and  Theorems \ref{esRYf6OWn5} \& \ref{lL3i4uiQDZ}.}
Note that Teutsch's result: 
\begin{enumerate}[\hspace{0.5cm}(i)]
\item  concerns integer-valued martingales which, in our notation, means that
 $\forall n,\ g(n)=0$;
\item states that there exists an integer-valued betting strategy $M$ such that 
for any countable family $(T_i)$ of integer-valued saving strategies, $\exists\ X\in\twome$ where
$M$ succeeds but also each $T_i$ fails to save on $X$.
\item his martingale $M$ mentioned in (ii) wagers at most 1 at each round, so in our terminology it is $g$-timid,
where $g$ is the constant zero function.
\end{enumerate}
Hence Teutsch's result implies the special case of Theorem \ref{esRYf6OWn5} where $\forall n\ g(n)=0$.
On the other hand, assuming $\forall n\ g(n)=0$, Theorem \ref{esRYf6OWn5} does not imply
Teutsch's result since neither (a) nor (b) of Theorem \ref{esRYf6OWn5} is required in Teutsch's result.
With regard to
Theorem \ref{lL3i4uiQDZ}, the hypothesis that $g$ is unbounded means that
it does not concern integer valued strategies. Furthermore 
the statement of Theorem \ref{lL3i4uiQDZ} without the hypothesis that 
$g$ is unbounded (which would make it applicable to the integer valued case)
would imply the negation of Teutsch's result, due to clause (iii) above.
Incidentally, the hypothesis in Theorem \ref{lL3i4uiQDZ} that $g$ is unbounded 
is used in establishing a universal vanishing non-decreasing upper bound $h$ for the wagers of $M$,
as shown in Table \ref{iHCIFQhir}, which is crucial in the argument of \S\ref{iPgHQCWM76}.
 
\subsection{Savings under fine granularity: proof of Theorem \ref{KOywTgVWA}}\label{T3Bi5M5HqJ}
Given nondecreasing $g$ suppose that 
$G$ be an integer strict upper bound of $2 + \sum_{n\in \Nat} 2^{-g(n)}$, let $M$ be
a supermartingale, and recall the statement of of Theorem \ref{KOywTgVWA}:
\begin{equation*}
\parbox{13cm}{there exists a $g$-granular supermartingale 
$N$, computable from $g,M$, such that $\lim_n S_N(X\restr_n)=\infty$ for each $X$ such that 
$\limsup_n M(X\restr_n)=\infty$.}
\end{equation*}
Without loss of generality, we assume $M(\sigma) \geq 1$ for all $\sigma \in \twolo$, because otherwise we may
use $M+1$ instead of $M$ in the argument.
We define the required supermartingale $N$ and its cover $\widehat{N}$ simultaneously, following a 
granular version of the
savings argument we used to justify \eqref{DxAECJj5D3}.
For any $\sigma\in \twolo$, 
consider the finite sequence $n_i, i<k$ defined inductively as follows:
$n_0 = 0$, and for each $i$ let $n_{i+1}$ be the least number (if such exists) such that  
$n_i<n_{i+1}$ and
$M(\sigma \restr_{n_{i+1}}) \ge 2 M(\sigma \restr_{n_i})$ for each $i$.
Then $k$ is the least number $i$ such that $n_i$ is undefined, and 
we may let $I(\sigma):=\{n_i\ |\ i<k\}$ and $l(\sigma ) := |I(\sigma)| - 1$.
We define $N,\widehat{N}$ by induction.

First, let $\widehat{N}(\lambda) = N(\lambda) = \left\lfloor G\cdot M(\lambda) \right\rfloor + 1$.
Then for each
$\sigma\in \twolo $, and each real $x$ let
$\intg(\sigma,x)$ be the largest integer multiple of $2^{-g(|\sigma|+1)}$ which is at most $|x|$, multiplied by
the sign of $x$. It follows that 
\begin{equation}\label{dHLi3lB5hJ}
\big|\ \intg(\sigma,x)-x\ \big |< 2^{-g(|\sigma|+1)}
\hspace{0.5cm}\textrm{for each $\sigma,x$.}
\end{equation}
We define the wager for $N$ on $\sigma$, based on the wager of $M$, but scaled by the fraction
$(\widehat{N}(\sigma) - l(\sigma))/M(\sigma)$ and rounded to the nearest granular value:
\begin{equation}\label{AffKOnFN3i}
w_N(\sigma) = \intg\left(\sigma, 
\frac{w_M(\sigma)\cdot \Big(\widehat{N}(\sigma) - l(\sigma)\Big)}{M(\sigma)}\right)
\end{equation}
as well as the values of $\widehat{N}(\sigma\ast i), N(\sigma\ast i)$ recursively, in terms of the values at $\sigma$:
\begin{equation}\label{FeiupPPNF2}
\parbox{11cm}{
$\widehat{N}(\sigma\ast 1) = \widehat{N}(\sigma)+ w_N(\sigma)$ 
\quad \ and \quad $\widehat{N}(\sigma\ast 0) = \widehat{N}(\sigma)- w_N(\sigma)$\\[0.2cm]
$N(\sigma\ast 1) = \widehat{N}(\sigma\ast 1)- l(\sigma)$  \quad and \quad 
$N(\sigma\ast 0) = \widehat{N}(\sigma\ast 0)- l(\sigma)$.}
\end{equation}
Clearly $N$ is $g$-granular and computable from $M$.
The intuition for the definition of $N$ is the same as the intuition in the argument 
for \eqref{DxAECJj5D3} that we discussed above, but
adapted to $g$-granular values. It remains to show that $\widehat{N}$ is the cover of $N$ and
$\lim_n S_N(X\restr_n)=\infty$ for each $X$ such that 
$\limsup_n M(X\restr_n)=\infty$.

\begin{lem}[Growth of $\widehat{N}$]\label{lem:tmg1}
For all $\sigma \in \twolo$, 
$M(\sigma)+l(\sigma) <\widehat{N}(\sigma)$.	
\end{lem}
\begin{proof}
By \eqref{dHLi3lB5hJ} and the definition of $w_N(\sigma)$ we have that for any $\sigma\in \twolo$:
\begin{equation*}
\Big| w_N(\sigma) -w_M(\sigma) \cdot \frac{\widehat{N}(\sigma) - l(\sigma)}{M(\sigma)} \Big|\leq 2^{-g(|\sigma|+1)}.
\end{equation*}
Then 
\[
\frac{\widehat{N}(\sigma\ast 1) - l(\sigma)}{M(\sigma\ast 1)}  
\ge \frac{\widehat{N}(\sigma) + w_M(\sigma) \cdot \frac{\widehat{N}(\sigma) - l(\sigma)}{M(\sigma)} 
- 2^{-g(|\sigma | +1)} - l(\sigma)}{M(\sigma) + w_M(\sigma)} 
= \frac{\widehat{N}(\sigma) - l(\sigma)}{M(\sigma)} - \frac{ 2^{-g(|\sigma | +1)}}{M(\sigma\ast 1)} 
\]
so
\begin{equation}\label{4gSAYexT49}
\frac{\widehat{N}(\sigma\ast i) - l(\sigma)}{M(\sigma\ast i)}\geq 
\frac{\widehat{N}(\sigma) - l(\sigma)}{M(\sigma)} - 2^{-g(|\sigma | +1)}
\end{equation}
for $i=1$. Similarly, under outcome 0 we have:
\[
\frac{\widehat{N}(\sigma\ast 0) - l(\sigma)}{M(\sigma\ast 0)}  
\ge \frac{\widehat{N}(\sigma) - w_M(\sigma) \cdot \frac{\widehat{N}(\sigma) - l(\sigma)}{M(\sigma)} 
- 2^{-g(|\sigma | +1)} - l(\sigma)}{M(\sigma) + w_M(\sigma)} 
= \frac{\widehat{N}(\sigma) - l(\sigma)}{M(\sigma)} - \frac{ 2^{-g(|\sigma | +1)}}{M(\sigma\ast 0)} 
\]
so \eqref{4gSAYexT49} also holds for $i=0$.
For $i\in \{ 0, 1\} $, if $I(\sigma\ast i) = I(\sigma)$, then we have $l(\sigma\ast i) = l(\sigma)$ and
\eqref{4gSAYexT49} gives:
\begin{equation}\label{EQ8jHLpurm}
\frac{\widehat{N}(\sigma\ast i) - l(\sigma\ast i)}{M(\sigma\ast i)} 
\ge \frac{\widehat{N}(\sigma) - l(\sigma)}{M(\sigma)} - 2^{-g(|\sigma\ast i|)}.
\end{equation}
If $I(\sigma\ast i) \neq I(\sigma)$, then 
$|\sigma\ast i | \in I(\sigma\ast i),\ l(\sigma\ast i) = l(\sigma)+1$ and
$M(\sigma\ast i) \ge 2^{l(\sigma\ast i)}\cdot M(\lambda)$.
Combining these facts with \eqref{4gSAYexT49}, we get:
\begin{equation}\label{hkgAzLLXbB}
\frac{\widehat{N}(\sigma\ast i) - l(\sigma\ast i)}{M(\sigma\ast i)} = 
\frac{\widehat{N}(\sigma\ast i) - l(\sigma)}{M(\sigma\ast i)} - \frac1{M(\sigma\ast i)}
\ge \frac{\widehat{N}(\sigma) - l(\sigma)}{M(\sigma)} - 2^{-g(|\sigma\ast i|)} - 2^{-l(\sigma)-1}.
\end{equation}
Inductively applying \eqref{EQ8jHLpurm} and \eqref{hkgAzLLXbB} for the cases  
$I(\sigma\ast i) = I(\sigma)$ or
$I(\sigma\ast i) \neq I(\sigma)$ respectively, we get:
\[
\frac{\widehat{N}(\sigma) - l(\sigma)}{M(\sigma)} 
\ge \frac{\widehat{N}(\lambda) - l(\lambda)}{M(\lambda)} - 
\sum_{n=1}^{|\sigma|} 2^{-g(n)} - \sum_{n=1}^{l(\sigma)} 2^{-n}
\ge G - \sum_{n\in \Nat} 2^{-g(n)} - 1 > 1
\]
which gives the required inequality.
\end{proof}

\begin{lem}[Properties of $N$ and $\widehat{N}$]
The function $\widehat{N}$ is a $g$-granular martingale and 
$N$ is a $g$-granular supermartingale; 
moreover $\widehat{N}$ is the cover of $N$.
\end{lem}
\begin{proof}
By the equations \eqref{FeiupPPNF2} in the definition of $N,\widehat{N}$ 
and since $l(\sigma\ast i) \ge l(\sigma)$ for $\ i\in \{ 0, 1\}$, we have 
$N(\sigma\ast i) = \widehat{N}(\sigma\ast i)- l(\sigma) \ge \widehat{N}(\sigma\ast i)- l(\sigma\ast i) > 0$,
which also shows that $\widehat{N}(\sigma)>0$ for all $\sigma$.
Given this fact, the equations \eqref{FeiupPPNF2} and the definition of the wager of a (super)martingale from 
\S \ref{W2nKOLMi5}
we get that:
\begin{itemize}
\item $N(\sigma\ast 1)+N(\sigma\ast 0)\leq 2N(\sigma)$ so $N$ is a supermartingale; 
\item $\widehat{N}$ is a martingale, $N(\sigma)\leq \widehat{N}(\sigma)$ for all $\sigma$ and
$N$, $\widehat{N}$ have the same wager $w_N$ given by \eqref{AffKOnFN3i};
\end{itemize}
Hence $\widehat{N}$ is the cover of $N$.
By  \eqref{AffKOnFN3i} the function $w_N$ is $g$-granular, so 
$\widehat{N},N$ are $g$-granular.
\end{proof}
Finally we verify that $N$ has the desired property. Suppose that 
$\limsup_n M(X\restr_n)=\infty$. 
It follows that $\lim_n l(X\restr_n)=\infty$. On the other hand
$N(X\restr_{n+1}) = \widehat{N}(X\restr_{n+1})- l(X\restr_n)$
and $N(X\restr_{n+1})>0$ for each $n$, so 
\[
\lim_n S_N(X\restr_{n+1})=
\lim_n \Big(\widehat{N}(X\restr_{n+1})-N(X\restr_{n+1})\Big)=
\lim_n l(X\restr_n)=\infty
\]
which concludes the proof of Theorem \ref{KOywTgVWA}.

\subsection{Strategy, notation, and outcome sequence: proof of Lemma \ref{fYSGaBvd6u}}\label{AVVUpDcdaX}
We view a saving strategy $T$ or supermartingale,
as a stochastic process on the underlying product space of binary outcomes.
We may also split each stage into
a {\em savings step},
 where $T$ can decrease producing a marginal saving, and a subsequent
{\em betting step}, when $T$ places a bet and the outcome is revealed 
(modifying the capital of $T$ accordingly).
%
In order to avoid overloaded notation,  in the following arguments we
specify a stage in the process and one of its steps 
(saving or betting) and talk about
the process $t$ (corresponding to $T$) 
at the beginning of the stage and step in question, denoting by $t'$ the capital
after the step has been completed (which is also the capital at the beginning of the next step).
Variables $g^{\ast},g^{+}$ denote the
{\em current granule} and the {\em next granule} of a stage in the game. Formally these
are the functions $\sigma\mapsto 2^{-g(|\sigma|)}, \sigma\mapsto 2^{-g(|\sigma|+1)}$ respectively,
determining the size of the required divisor of the current capital and wager respectively at each stage, based
on the function $g$.
Recall the generalized notions of  {\em quotient} and {\em remainder} in Definition \ref{zY3U5A7rVM}.
We let $w$ be the random variable corresponding to $\sigma\mapsto w_T(\sigma)$
and let $q,r$ be the quotient and remainder of the division of $t$ by $m$, with 
$q', r'$ denoting their values after a step in the process has been completed. Hence
$t=q\cdot m+r$, where $r<m$
and $q$ is an non-negative integer.  
These notational conventions are summarized in
Table \ref{ABuEidRZu} 
where $\mathbb{Z}^{+}$ denotes the set of non-negative integers and $g^{\ast}\cdot\mathbb{Z}^{+}$
the set of non-negative integer multiples of $g^{\ast}$. 

\begin{table}
\colorbox{black!5}{\arrayrulecolor{white!20!black} 
  \begin{tabular}{lcc}
{\small\em Notion} &{\small\em Variable} & {\small\em Value at  $s$}  \\[0.1ex]\cmidrule[0.5pt]{1-3}
{\small Current granule} &{\small $g^{\ast}$} 	 & {\small $2^{-g(s)}$ } \\[0.5ex]
 {\small Next granule} &{\small $g^{+}$}  &  {\small $2^{-g(s+1)}$}\\[0.5ex]
{\small Wager of $T$} &{\small $w$} & {\small  $\in g^{+}\cdot \mathbb{Z}$}\\[0.5ex]
{\small Next outcome } &{\small $x$} & {\small binary}\\[0.5ex]
\end{tabular}}\centering\hspace{0.2cm}
\colorbox{black!5}{\arrayrulecolor{white!20!black} 
\begin{tabular}{lccc}
{\small\em Notion}&{\small\em Before / After} &\hspace{0.2cm}  &  {\small\em Type}  \\[0.1ex]\cmidrule[0.5pt]{1-4}
{\small Capital of $T$} &{\small $t,t'$} &	\hspace{0.2cm}  & {\small $g^{\ast}\cdot \mathbb{Z}^{+}$} \\[0.5ex]
 {\small Capital of  $M$} &{\small $m,m'$}	& \hspace{0.2cm}  &  {\small $g^{\ast}\cdot \mathbb{Z}^{+}$}\\[0.5ex]
 {\small Quotient $\floor{t/m}$} &{\small $q,q'$} & \hspace{0.2cm} & {\small $\mathbb{Z}^{+}$}\\[0.5ex]
 {\small Remainder $t/m-\floor{t/m}$} &{\small $r,r'$} & \hspace{0.2cm} & {\small $g^{\ast}\cdot \mathbb{Z}^{+}$}\\[0.5ex]
\end{tabular}}\centering
\caption{Parameters for the proof of Lemma \ref{fYSGaBvd6u}}\label{ABuEidRZu}
\end{table}


Recall the statement of Lemma \ref{fYSGaBvd6u}, which assumes that 
$g$ is nondecreasing:
\begin{equation*}
\parbox{13.5cm}{if $\sum_{n} 2^{-g(n)}=\infty$,
there exists a $g$-timid martingale $M\leq_T g$ 
such that for every  $g$-granular supermartingale $T$
there exists $X$ with $\lim_n M(X\restr_n)=\infty$ and $\lim_n S_{T}(X\restr_n)<\infty$.}
\end{equation*}
The martingale $M$ of Lemma \ref{fYSGaBvd6u} is the strategy that starts with capital $2^{-g(0)}$ and
at each stage $n+1$ bets $2^{-g(n+1)}$ on outcome 1 (unless
its current capital is less than this, in which case it does not bet). Formally,
$M(\sigma\ast j) =M(\sigma) + (-1)^{j+1}\cdot 2^{-g(|\sigma|+1)}$ if
$M(\sigma) \ge 2^{-g(|\sigma|+1)}$, and $M(\sigma)$ otherwise.
Using our simplified notation, $M$ can be written as the process $m$ which at every betting step
is determined by:
\begin{equation*}
m'=\left\{\begin{array}{ll}
m + g^{+}&\textrm{if outcome is 1 and $m\geq g^{+}$;}\\
m - g^{+}&\textrm{if outcome is 0 and $m\geq g^{+}$;}
\end{array}\right\}.\ \ \ 
\textrm{If $m< g^{+}$ then $m'=m$.} 
\end{equation*}
%

{\bf Analysis of transitions and outcome sequence $X$.}
By Lemma \ref{nXAZQtLDNS} we may assume that 
$\sigma\mapsto T(\sigma)$ is $g$-granular.
For Lemma \ref{fYSGaBvd6u} we wish to construct an infinite sequence of outcomes $X$
along which $m$ diverges to infinity (\ie $M(X\restr_n)\to\infty$ as $n\to\infty$)
while the given $t$ accumulates a finite amount of savings along $X$ (\ie $\lim_n S_T(X\restr_n)$
is finite). To this end we ensure that along $X$ the ratio
$t/m$ is non-increasing, \ie
\begin{equation}\label{Z1i12nlwpS}
t'/m'\leq t/m
\hspace{0.7cm}\textrm{at each saving or betting step  along $X$.}
\end{equation}
%
For a saving step, \eqref{Z1i12nlwpS} follows from the fact that 
$m'=m$ and $t'\leq t$. For betting steps along $X$,
we have to choose the outcomes appropriately.
Under the 1-outcome, $t'=t+w$ and $m'=m+g^{+}$, so it suffices that
\[
(t+w)/(m+g^{+})\leq t/m\Leftrightarrow
w \leq g^{+}\cdot q + g^{+}\cdot r/m
\Leftrightarrow
w \leq g^{+}\cdot q 
\]
since $r<m$, $q$ is an integer
and $w$ is an integer multiple of $g^{+}$.
Under 0-outcome we have $t'=t-w$ and $m'=m-g^{+}$, so
it suffices that
$(t-w)/(m-g^{+})< t/m\Leftrightarrow
w > g^{+}\cdot q + g^{+}\cdot r/m
\Leftrightarrow w > g^{+}\cdot q$, 
%
and note that we need strict inequality in the 0-outcome in order to get the equivalence.
Hence
\begin{equation}\label{dWqsjR6bsC}
\parbox{13cm}{if we follow the rule {\em 1-outcome if 
$w \leq g^{+}\cdot q'$ and 0-outcome if $w > g^{+}\cdot q'$}
then $t'/m'\leq t/m$ at each step/stage;
in the case of 0-outcome in a betting step, $t'/m'< t/m$.}
\end{equation}
Based on \eqref{dWqsjR6bsC}, let $x$ be the outcome chosen at each betting step 
(where $t, m,g^{+},q$ are defined) and define:
\begin{equation}\label{1XGEpgJS}
\textrm{$x =1$ if $w \leq g^{+}\cdot q$; 
\hspace{0.3cm} and \hspace{0.3cm}  $x=0$ if $w > g^{+}\cdot q$.}
\end{equation}
This recursive equation defines the sequence $X$ of outcomes for Lemma \ref{fYSGaBvd6u}.
Note that the present analysis rests
on the hypothesis that $m$ remains positive, which is a property that will be verified.

\subsection{Verification of fine granularity strategy: concluding the proof of Lemma \ref{fYSGaBvd6u}}\label{VEtaAeoOkI}
Given $M$ of \S\ref{AVVUpDcdaX} and a $g$-granular saving strategy $T$, we define 
$X$ by \eqref{1XGEpgJS} and it remains to show that $\lim_n M(X\restr_n)=\infty$ and
$\lim_n S_{T}(X\restr_n)<\infty$.
First we show that bankruptcy is avoided along $X$.
\begin{lem}[$M$ is never bankrupt along $X$]\label{lem:mwd}
At all stages  along $X$ we have $m \ge g^{\ast}$; in the case of a saving step we have
$m'=m \ge g^{\ast}$ and in the case of a betting step $m' \ge g^{+}$.
\end{lem}
\begin{proof}

Recall that at each stage and step, $m$ denotes the value of $M$ along $X$ at the beginning of the given step, while
$m'$ denotes the  value of $M$ at the end of the given step. 
At each betting step $m'$ is also the value of $M$
along $X$ at the beginning of the next stage, so it suffices to prove 
by induction on the stages and steps that $m \ge g^{\ast}$ at each step.
At stage 0 we have $m=2^{-g(0)}=g^{\ast}$. Inductively suppose that
$m \ge g^{\ast}$ at the start of some stage, so $m'=m\geq g^{\ast}$ at the saving step. In order to 
complete the induction step, it suffices to show that at the next betting step we have
$m'\ge g^{+}$ (which is equivalent to $m \ge g^{\ast}$ referenced at the next stage).

If $g^{+}<g^{\ast}$, by the definition of $M$
we have $m'\geq m-g^{+}>m-g^{\ast}\geq 0$ so $ m'>0$ and since $m'$ is $g$-granular,
we have that $m'\geq g^{+}$ as required.
If $g^{+}=g^{\ast}$ and $m>g^{\ast}$, by the same argument we get
$m'\geq g^{+}$ as required. So it remains to examine the case where  $m=g^{\ast}$
at the betting step in question. In this case it suffices to show that the outcome chosen by $X$ will by 1,
so that $M$ increases its capital. Since  $m=g^{\ast}$ it follows that $m$ divides $t$, \ie $r=0$,
so $t=q\cdot m=q\cdot g^{\ast}=q\cdot g^{+}$. Moreover $|w|\leq t$, since $T$ cannot bet
more than its current capital, so $w\leq q\cdot g^{+}$. According to the definition of $x$ in  
\eqref{1XGEpgJS} it follows that the chosen outcome is 1, as required, which concludes
the induction step and the proof of the lemma.
 \end{proof}

\begin{lem}[Monotonicity of ratios]\label{lem:lwd}
At each stage and step along $X$ we have $q'\leq q$; if $q'=q$ then $r'\leq r$.
At a betting step where $q'=q$ and
the 0-outcome is chosen, $r'<r$.
\end{lem}
\begin{proof} 
By \eqref{1XGEpgJS} and \eqref{dWqsjR6bsC}
we have $t'/m'\leq t/m$ so
$q'+r'/m'\leq q+r/m$, and since 
$q,q'$ are integers, $r'<m'$ and $r<m$, it follows that
$q'\leq q$.
For the second part, assume that $q'=q$.
Under outcome 1
we have $m'=m+g^{+}$, $t'=t+w$, $w \leq g^{+}\cdot q$;
under outcome 0 we have 
$m'=m-g^{+}$, $t'=t-w$, $w > g^{+}\cdot q$. Hence:
\[
\begin{tabular}{cc}
\textrm{\em Outcome 1:}&
$r'=t'-q\cdot (m+g^{+})=t+w-qm-qg^{+}\leq t-qm=r$\\[0.1cm]
\textrm{\em Outcome 0:}&
$r'=t'-q\cdot (m-g^{+})=t-w-qm+qg^{+}< t-qm=r$
\end{tabular}
\]
which concludes the proof of the lemma.
\end{proof}

\begin{lem}[Marginal savings and 0-outcomes]\label{16SbRAgwbi}
Along $X$, after some stage $q,q'$ remain constant and $r'\leq r$.
The marginal savings of $T$ at such a sufficiently large saving step along $X$
is $r-r'$ and in the case of a sufficiently large betting step where the 0-outcome is chosen we have $r-r'\geq g^{+}$.
\end{lem}
\begin{proof}
The first statement of the lemma follows from Lemma \ref{lem:lwd} since $q,q'$ are
non-negative integers. By the divisions $t=q\cdot m+r$, $t'=q'\cdot m'+r'$ and the
fact that $q'=q, m'=m$ at saving steps of any sufficiently large stage, 
it follows that $t-t'$, \ie the marginal savings of $t$ at this stage,
equals $r-r'$. Finally  $r-r'\geq g^{+}$ for the case of a 0-outcome in a 
sufficiently large betting step follows from 
Lemma \ref{lem:lwd} and the fact that $r,r'$ are integer multiples of $g^{+}$.
\end{proof}
\begin{lem}[Total savings of $T$ and growth of $M$]\label{lem:cbd}
Consider a stage after which  $q,q'$ remain constant along $X$, and let $r_0$ be
the value of $r$ at the start of that stage along $X$.
Then the remaining savings of $T$ along $X$ are at most $r_0$.
Moreover $m\to\infty$ along $X$.
\end{lem}
\begin{proof}
The total savings of $T$ along $X$ equals the sum of its marginal savings along $X$. Hence
the first part of the lemma follows from Lemma \ref{16SbRAgwbi} and the fact that $r$ remains
non-negative throughout the process, as well as non-decreasing after a stage where
 $q,q'$  have reached a limit.
For the second part, recall that $m'-m=g^{+}$ 
at betting steps where outcome 1 is chosen. 
At betting steps where outcome 0 is chosen we have 
$m'<m$ 
and by Lemma \ref{16SbRAgwbi},
$-(m'-m)=g^{+}\leq r-r'$. 
At saving steps, $m'=m$.
Hence 
if $s_0$ is the stage mentioned in the statement of the present lemma, and $r_0$ is the value of $r$ at
that stage, we have
$M(X\restr_s)\geq \sum_{s\geq s_0} 2^{-g(s+1)} -r_0$, which shows that
$m\to\infty$ along $X$.
\end{proof}

\subsection{Bounded or timid saving adversary: proof of Theorem \ref{esRYf6OWn5}}\label{kqkjmaDuX}
Given $g$ with $\sum_i 2^{-g(i)}=\infty$, Theorem \ref{esRYf6OWn5} asserts the existence of a 
$g$-timid and strategy $M\leq_T g$ such that for any countable
family $(T_i)$ of $g$-granular saving strategies:
\begin{equation}\label{xYVnzguxPT}
\parbox{13.5cm}{there exists $X$ such that
 $\limsup_n M(X\restr_n)=\infty$ and $\limsup_n S_{T_i} (X\restr_n)<\infty$ for each $i$, provided that
 one of the following holds: 
(a) $(T_i)$ is bounded; or (b) each $T_i$ is $g$-timid.}
\end{equation}
As noted in the first part of \S\ref{MEchtencE}, part (a) of Theorem \ref{esRYf6OWn5} follows from 
Lemma \ref{fYSGaBvd6u}, so it remains to prove the case where each $T_i$
is $g$-timid. 
We adopt the notation of
\S\ref{AVVUpDcdaX}, also summarized in Table \ref{ABuEidRZu}, 
and let $t_i$ denote the process associated with $T_i$.
Also let $(c_i)$ be a sequence of positive integers such that for each $i$,
the wager of $T_i$ on strings of length $k$ is $<(c_i-1)\cdot 2^{-g(k)}$.
We let $M$ be the the martingale of Lemma \ref{fYSGaBvd6u}, which was formally defined in \S\ref{AVVUpDcdaX}.

{\em Intuitive outline.}
In the special case where the $T_i$ are timid, note that \eqref{xYVnzguxPT} is a
`universal' version of Lemma \ref{fYSGaBvd6u}, in the sense that it requires all pairs $(M,T_i)$ 
to meet the requirement `$M$ succeeds betting and $T_i$ fails saving', as opposed to a single $(M,T)$
in Lemma \ref{fYSGaBvd6u}. In this light, our
plan for the proof of case (b) of \eqref{xYVnzguxPT} 
is to implement the argument of \S\ref{AVVUpDcdaX} concurrently for all $T_i, i\in\Nat$.
In order to avoid interference of the actions of $M$ with respect to different members of the list
$T_i, i\in\Nat$, we need to {\em nest} the requirements through the use of the
following hierarchy of  divisions based on Definition \ref{zY3U5A7rVM}.\footnote{This nesting argument
is based on the proof of the saving paradox of 
\citet{Teu14agCWGnp} in \cite[Theorem 14]{Peretzwager}.}

{\em Hierarchy of nested divisions.}
Let $m_0:=m$, consider the integer quotient of $q_0:=\floor{t_0/m_0}$ and 
consider the remainder $r_0=t_0-q_0\cdot m_0$. 
For each $i>0$, assuming that $m_{i-1},r_{i-1}$ are defined and 
$c_i\cdot r_{i-1}+i<m_{i-1}$, define inductively
\begin{equation}\label{AEXRV6b1Ro}
m_i=m_{i-1}-c_i\cdot r_{i-1}-i
\hspace{0.5cm}\textrm{and}\hspace{0.5cm}
q_i:=\left\lfloor t_i/m_i\right\rfloor
\hspace{0.5cm}\textrm{and}\hspace{0.5cm}
r_i=t_i-q_i\cdot m_i.
\end{equation}
Note that $m$ is a random variable indicating the capital of the strategy $M$ along the tree of binary outcomes,
while  $m_i$ is simply a parameter (also random variable) obtained recursively from $m$ and $(t_j)$ 
after $i$ many iterations of the recursive definition \eqref{AEXRV6b1Ro}. 
In particular, $m_i, i>0$ is not  a martingale or wager related to $M$,
but simply a function of  $m, t_j, j<i$.
On the other hand, for each $i$ we let  $w_i$ denote the wager that $T_i$ bets on outcome 1
at the present step.

{\em Technical outline.}
For (b) of \eqref{xYVnzguxPT} we need define $X\in\twome$ which satisfies two conditions:
$M$ succeeds ($m\to\infty$) along $X$  and each $T_i, i\in\Nat$ fail to save along $X$. The parameters $m_i, t_i,r_i$
of the nested divisions will be used in order to determine the bits of $X$.
The nested divisions will ensure
that $m>0$ at all states of the process where these nested divisions are considered; this crucial property 
would be threatened if we directly divided $t_i$ by $m$ and 
implemented the
argument of \S\ref{AVVUpDcdaX} concurrently for all $T_i, i\in\Nat$ without nesting.
The combination of:
\begin{itemize}
\item  the hierarchical relationship of parameters $m_i, t_i,r_i$
\item the assumption that   $T_i, i\in\Nat$ are timid
and the property that $m>0$ along $X$
\end{itemize}
 allow the application of
the argument of \S\ref{AVVUpDcdaX} for the 
 establishment that
$m\to\infty$ along $X$. In order to establish that $T_i$ saves at most a finite amount along $X$
we ensure that the nested devisions eventually stabilize and provide a bound on the future savings of $T_i$
in the sense that for each $i$: 
\begin{itemize}
\item after a sufficiently large bit along $X$ quotient $q_i$ becomes constant and $r_i$ becomes nonincreasing;
\item at any larger position along $X$, the remaining savings of $t_i$ along $X$ are bounded by $r_i$.
\end{itemize}
We use $t_i, t_i'$ to denote the values of $T_i$ at the beginning and the end
of a transition or stage of the process, and similar notation applies to
$q_i,r_i,m_i,m$.  For each $i>0$ we say that $m_i$ {\em is defined} at a certain state,
denoted by $m_i\de$, if
$m_{i-1}$ is defined and  $i+c_i\cdot r_{i-1}<m_{i-1}$.
It follows that if $m_n\de$ then $m_j< m_i-r_i$ for all $j\leq i\leq n$.
We say that {\em $t_i$ requires attention} at some state if $m_i$ is defined and  
$w_i\neq q_i\cdot g^+$.

{\bf Construction for (b) of Theorem \ref{esRYf6OWn5}.}
At stage $s+1$ let $i\leq s$ be the least such that $t_i$  requires attention.
If no such $i$ exists, choose outcome 1. Otherwise, if 
$w_i\leq q_i\cdot g^+$ choose outcome 1 and 
if $w_i > q_i\cdot g^+$
choose outcome 0.
Let $X$ be the binary sequence determined by this choice of outcomes.

{\bf Verification  for (b) of Theorem \ref{esRYf6OWn5}.}
We use the argument of \S \ref{VEtaAeoOkI}  to show that,
along $X$, for all $i$,
\begin{equation}\label{OgJAr3Mv7w}
\textrm{$\exists $ stage after which:\ \ }
\left\{\ \ \parbox{7.7cm}{%
\begin{enumerate}[(A)]
\item $m_i\de\ \geq i$, $q_i$ is constant and $r_i$ is nonincreasing;
\item if some $j\leq i$ requires attention then $r_i'<r_i$.
\end{enumerate}
%
}\ \ \right\}
\end{equation}
For $i=0$, every time $t_i$ requires attention the outcome will be chosen
by the same rule as in the proof of Lemma \ref{fYSGaBvd6u}. At stages
where $t_0$ does not require attention, 
if $m_0=r_0+g^{+}$ then $m_j=g^{+}$ and hence 
$w_j\leq t_j=q_j g^{+}$ for all $j$ such that $m_j\de$.
This means that at stages where
$t_0$ does not require attention, 
and $m_0=r_0+g^{+}$, outcome 1 will be chosen, ensuring that $m_0'>m_0$.
Hence at  stages where
$t_0$ does not require attention, 
$r_0'=r_0$ and $q_0'=q_0$. 
With these observations,
the argument of \S\ref{VEtaAeoOkI}
proves \eqref{OgJAr3Mv7w} for $i=0$.

Now inductively let $k>0$ and assume that \eqref{OgJAr3Mv7w} holds for all $i<k$.
In order to prove  \eqref{OgJAr3Mv7w} for $i=k$,
let $s_0$ a stage along $X$ which is larger than the stage mentioned in  \eqref{OgJAr3Mv7w}
for each $i<k$. 
After $s_0$ along $X$, for each $i<k$ we have $m_i\de$ and $r_i$ is non-increasing.
First we show that at some $s_1\geq s_0$ we have $m_k\de$, \ie 
$m_{k-1}>c_k\cdot r_{k-1}+k$. For a contradiction, suppose that this is not the case, so
for all $s\geq s_0$ and all $j\geq k$ we have $m_j\un$ and $t_j$ does not require attention.
By the monotonicity of $r_{k-1}$, this means that $m_{k-1}$ is bounded above.
Moreover, after $s_0$ along $X$, every time outcome 0 is chosen, some $t_j$ with 
$j<k$ requires attention, so by the induction hypothesis $r_{k-1}'<r_{k-1}$.
By the monotonicity of $r_{k-1}$, the fact that it is bounded below by 0 and the granularity of 
$r_{k-1}$, it follows that the sum of all $g^+$ along $X$ at stages after $s_0$ where outcome
0 is chosen is finite. Hence the downward variation of $m$ is finite, 
and since the sum of all $g^+$ along $X$ after $s_0$ is infinite (on the assumption of
coarse granularity) it follows that the upward variation of $m$ along $X$ is infinite.
Hence $m\to\infty$ along $X$, and since $r_{i}, i<k$ are non-increasing along $X$, it follows that
$m_{k-1}\to\infty$ along $X$, which contradicts the assumption that $m_k$ is not defined after
$s_0$ along $X$. This contradiction shows that at some $s\geq s_0$ along $X$,
we have $m_{k-1}>c_k\cdot r_{k-1}+k$ so $m_k\de$.
Let $s_1$ be the least such stage.

Next, we show that $m_k\de$ at all $s\geq s_1$ along $X$.
The only reason why this may fail is that 
$m_{k-1}'\leq c_k\cdot r_{k-1}'+k$ at some 
$s\geq s_1$ along $X$ which, by the induction hypothesis, implies that 
$m_{k-1}= c_k\cdot r_{k-1}+k+g^+$, outcome 0 is chosen by the construction along $X$,
and $r_{k-1}'\geq r_{k-1}$. At such a stage $r_k=0$ and $m_k=g^+$ so $w_k\leq t_k= q_k g^+$ 
and  $ q_k g^-w_k\geq 0$, so it is not possible that outcome 0 is chosen by $t_k$ requiring attention.
On the other hand, if some $j<k$ requires attention, by the induction hypothesis we have
$r_{k-1}'< r_{k-1}$ which also disagrees with the above conditions. Finally if no $t_j, j\leq k$ requires
attention, outcome 1 is chosen at the given stage along $X$.
This concludes the proof that for each  $s\geq s_1$ along $X$ we have  $m_k\de$
and, by definition, $m_k\geq k$.

In order to show (A) of \eqref{OgJAr3Mv7w}, 
we show that after stage $s_1$ along $X$ we have
$q_k'< q_k\ \vee\ (q_k'= q_k \wedge\ r_k'\leq r_k)$.
At saving steps we have $t_k'=t_k$ and $m_k'\geq m_k$ so 
$q_k'< q_k\ \vee\ (q_k'= q_k \wedge\ r_k'\leq r_k)$. If at some betting step after $s_1$
some $j<k$ requires attention, then by the inductive hypothesis $r_{k-1}'<r_{k-1}$.
In this case, $m_k'-m_k\geq c_k\cdot g^+ -g^+> w_k$. 
Hence if  $q_k=0$ we have $q_k'=0$ and  $r_k'< r_k$; 
if $q_k>0$ then $w_k< q_k (m_k'-m_k)$ so either $q_k'<q_k$ or 
$q_k'=q_k\ \wedge\ r_k'<r_k$.
If no $j<k$ requires attention at the betting step then
$r_{k-1}'=r_{k-1}$. In this case, if $w_k\neq q_k\cdot g^+$ the construction 
will choose the outcome so that $q_k'< q_k\ \vee\ (q_k'= q_k \wedge\ r_k'\leq r_k)$;
otherwise $q_k'=q_k$ and $r_k'=r_k$.
This completes the proof that 
at each betting or saving step after $s_1$ along $X$, we have 
$q_k'< q_k\ \vee\ (q_k'= q_k \wedge\ r_k'\leq r_k)$. Hence $q_k$ reaches a limit
at some stage $s_2\geq s_1$ and this shows property (A) of \eqref{OgJAr3Mv7w}
for $i=k$.

Finally, to conclude the inductive step for the proof of \eqref{OgJAr3Mv7w}, we show that
after stage $s_2$ along $X$, if some $j\leq k$ requires attention then $r_k'<r_k$.
In the case where $j<k$, given that $q_k'=q_k$, 
 this follows by the expression $m_k=m_{k-1}-c_k r_{k-1}-k$,
the induction hypothesis which gives $r_{k-1}'<r_{k-1}$, 
and the fact that $|w_k|< (c_k-1)\cdot g^+$. In the case that $j=k$ the construction
choses the outcome so that $r_{k-1}'<r_{k-1}$, provided that $q_k'=q_k$ (which we also have
by the choice of $s_2$). This completes the inductive proof of \eqref{OgJAr3Mv7w}.

Clause of (A) of \eqref{OgJAr3Mv7w}  implies that $m\to\infty$ along $X$, so 
 $M$ is successful along $X$.
For the proof of clause (b) of Theorem \ref{esRYf6OWn5},
it remains to show that each $T_i$ saves at most a finite capital along $X$. Given $i\in\Nat$ let
$t_0$ be a stage with the properties of  \eqref{OgJAr3Mv7w}. At each saving step after stage $t_0$,
strategy $T_i$ saves exactly $r_i-r_i'$. Since $r_i$ is nonincreasing after $t_0$ and is lower bounded by 0,
it follows that after $t_0$ strategy $T_i$ can save at most $r_i[t_0]$. This completes the proof that $X$ meets
the properties of clause (b) of Theorem \ref{esRYf6OWn5}.

\section{Countable savings cover via hedging: proof of Theorem \ref{lL3i4uiQDZ}}\label{iPgHQCWM76}
Assuming that $g$ is nondecreasing, unbounded and $M$ is $g$-timid betting strategy,
Theorem \ref{lL3i4uiQDZ} asserts the existence of a
countable family $T_i, i\in\Nat$ of saving strategies such that:
\begin{enumerate}[\hspace{0.5cm}(i)]
\item $(T_i)$ is computable from $M$ with wagers  integer multiples of the corresponding wagers of $M$;
\item for each $X$ where $\limsup_n M(X\restr_n)=\infty$ there exists $i\in\Nat$ such that 
$\lim_n S_{T_i} (X\restr_n)=\infty$.
\end{enumerate}
%
Given $M,g$ as above, we construct saving strategies $T$ and $N_{\sigma}, \sigma\in I$,
where $I$ is a  certain set of binary 
strings that will be defined  below. The constructed $T, I, N_{\sigma}, \sigma\in I$
will be computable in $M,g$.
Here $T$ will be the main strategy, which follows $M$ closely, but not exactly, and attempts to
accumulate savings along each path where $M$ appears to succeed.

{\em Cycles and sub-cycles.}
Strategy $T$ works in cycles which repeatedly {\em start}, {\em pause},  {\em resume} or 
{\em end} along any possible path of binary outcomes. The parts of a cycle along a path, between a
resumption (or start) and a pause (or end) are called {\em sub-cycles}, and can be thought off as
the `rounds' of the cycle. 
A new (sub)cycle along a path can only start if the previous one has ended. 
Hence any path of outcomes can be partitioned into intervals $[\sigma,\tau]$ which are one of the following:
\begin{itemize}
\item {\em Sub-cycle interval:} 
a sub-cycle that started (or resumed) on $\sigma$ and ended (or paused) on $\tau$;
\item {\em Neutral interval:} a neutral interval which began straight 
after the end (or pause) of a sub-cycle, and ended
just before the start (or resumption) of a sub-cycle;
\end{itemize}
Intuitively, as we elaborate below, each `round' or sub-cycle starts 
with the goal of making either $T$ or an associated
strategy $N_{\sigma}$ save; which of these 
is achieved depends on the bets of $M$ along the given path.

{\bf Hedging against risky bets.}
A cycle of $T$ along a path of outcomes is when $T$ attempts to generate savings through its bets, while
during neutral intervals it does not save. In order to achieve savings, during cycles $T$ betting deviates 
slightly from the betting of $M$, with respect to the amount of the wagers, but following the outcome preference
of $M$.
The set $I$ which indexes the family $N_{\rho}$ consists of the strings 
where a new cycle (and its sub-cycles) of $T$ starts. 
A cycle starting at $\rho$ and running along a path extending $\rho$,
possibly with pauses, is associated with strategy $N_{\rho}$, 
which may be thought off as a backup or a {\em hedge against} $T$ during an interval
where its bets are more risky (deviating from the safe strategy $M$).
A cycle of $T$ runs on the same stages as the construction, measured in terms of the number of outcomes
revealed, or the number of bets placed. Along any path, a sub-cycle starting at $\sigma$  
may end in one of two possible ways: 
\begin{itemize}
\item {\em $T$-success of a sub-cycle:} 
strategy $T$ is in a position to save a certain amount, in which case the associated 
cycle is canceled and $N_{\rho}$ will not initiate further sub-cycles along any path extending $\sigma$.
\item {\em $T$-failure of a sub-cycle:} strategy $T$ is not in a position to save, but the associated backup 
strategy $N_{\rho}$ saves a certain amount. In this case $N_{\rho}$ may run further
sub-cycles along extensions of $\sigma$.
\end{itemize}
The {\em life interval} $[\rho,\tau]$  of $N_{\rho}$ along a path $Z$ (where $\tau$ may be the entire path $Z$) 
is determined by the starting stage $\rho$ and the ending stage $\tau$, leaving the interval open-ended in case
$N_{\rho}$ is never canceled along the path. During its life interval,  $N_{\rho}$ will
always bet, even during neutral intervals, but will only save upon each $T$-failed completion of its sub-cycles.
We will ensure one of the following outcomes along any path $X$:
%
\begin{enumerate}[\hspace{0.5cm}(a)]
\item $\limsup_n M(X\restr_n)<\infty$: in this case only finitely many sub-cycles are initiated along $X$;
\item infinitely many cycles $N_{\rho}$ initiated and ended along $X$: in this case $T$ gradually saves successfully;
\item there is a last cycle $N_{\rho}$ initiated along $X$, which never ends and generates
infinitely many sub-cycles along $X$, each of them ending in failure for $T$: in this case $N_{\rho}$
successfully saves along $X$. 
\end{enumerate}
{\bf Backup strategies and parameters.}
The backup strategies $N_{\rho}$ will be defined depending on the type of the stages.
They bet identically to $M$ in neutral intervals, while in sub-cycles intervals they bet
anti-symmetrically (\ie bet on different outcome) to $M$ with amplified wagers.
One backup strategy $N_{\rho}$ might end 
along a path of outcomes at some stage $\eta$, in which case we will drop it 
by setting $N_{\rho}(\sigma) = N_{\rho}(\eta)$ for all $\sigma \succ \eta$.
Then we will initiate a new backup strategy $N_{\eta}$ which will receive a initial capital of $M(\eta)$
and only start to bet at $\eta$.
In this way at any stage $\sigma$ there is exactly one active backup strategy.
Thus, for the construction at every stage we only need to specify the wager of the current active backup strategy, 
while for all other backup strategies their wagers are 0.
If $N_{\rho}$ is the active backup strategy at stage $\sigma$,
we define the index of the stage $\sigma$ as $i_{\sigma} = \rho$.
Our construction will ensure that $i_{\sigma}\preceq i_{\tau}$ when $\sigma\preceq\tau$.
For simplicity we: 
\begin{equation}\label{LVSHDGcf4N}
\parbox{13cm}{often omit the subscript $\rho$, 
when it is clear which backup strategy is the active one}
\end{equation}
for the current stage.
The difference $r_{\sigma}=r(\sigma)=T(\sigma)-M(\sigma)$
plays a special role in the argument, and can be thought off in the context of the argument in 
\S\ref{AVVUpDcdaX} as the remainder of the division of $T$ by $M$.
The bets of $N$ during sub-cycles 
depend on parameter $c_{\sigma}$ which is updated after each sub-cycle and serves as a marker of the value of the $r$ 
at the starting stage of a sub-cycle. A sub-cycle ends when
$r_{\sigma}$ escapes $(c_{\sigma}/ 2,c_{\sigma}+1)$; it is $T$-successful if $r_{\sigma}\geq c_{\sigma}+1$ and $T$-unsuccessful
if $r_{\sigma}\leq c_{\sigma} / 2$. 
Hence, intuitively speaking, the strategies guess whether $r_{\sigma}$ is going to escape $(c_{\sigma}/ 2,c_{\sigma}+1)$ 
from the right or the left end, with $T$ betting on the first outcome and $N$ betting on the latter.
Let $w_{\sigma}$ be the wagers of $M$. Since $M$ is $g$-timid and $g$ is nondecreasing and unbounded, 
there exists 
$M$-computable non-increasing $h:\Nat\to\mathbb{Q}$ which tends to 0 and such that
$w_{\sigma}\leq h(|\sigma|)$ for each $\sigma$. Using this property of $h$, \ie the $g$-timidness of $M$,
we will ensure that
once $r_{\sigma}$ escaped from $(c_{\sigma}/ 2,c_{\sigma}+1)$, it is still in $(c_{\sigma}/ 4,c_{\sigma}+2)$.
This is crucial in ensuring that $T, N_{\rho}$ are supermartingales.
These parameters are summarized in Table \ref{iHCIFQhir}.

\subsection{Hedging strategy for bold play: formal construction for Theorem \ref{lL3i4uiQDZ}}\label{CHwOEWwFd}
Given $M$ as in the statement of Theorem \ref{lL3i4uiQDZ}, we
define the saving strategy $T$, the set $I$ 
of strings/stages where new cycles are initiated, 
and the family $N_{\rho}, \rho\in I$ of
backup saving strategies. In terms of the index function $\sigma\mapsto i_{\sigma}$ discussed above (formally defined below) we may set:
\begin{equation*}
I=\{i_{\rho}\ |\ i_{\rho}\neq i_{\hat{\rho}}\}
\hspace{0.5cm}\textrm{where $\hat{\rho}$ denotes the predecessor of $\rho$.}
\end{equation*}
If $i_{\sigma} = \rho $, then we say $N_{\rho}$ is the {\em active} backup of $T$ at stage $\sigma$.
In the spirit of the informal discussion at the beginning of \S\ref{iPgHQCWM76},
stages will be inductively classified as {\em neutral stages} or {\em sub-cycle stages}.
The last stage of a sub-cycle is also called as (sub-cycle) {\em ending  stage}.
Our strategies only save on the ending stages.
For simplicity in the formal construction we use 
the conventions of \S\ref{AVVUpDcdaX} and for each  stage $\sigma$: 
 \begin{list}{$\bullet$}{\leftmargin=1em \itemindent=1em}
\item each ending stage $\sigma$ is divided into two steps, the betting and the saving step.
\item the value  $r$ at the end  betting steps is denoted by $r^0_{\sigma}$,
while $r_{\sigma}$ refers to the  end of saving steps.
\item we only specify the wager of the current active backup strategy, 
as for all other backup strategies their wagers are 0;
the savings  $s_n(\sigma), s_t(\sigma)$ are always assumed to be 0, unless
explicitly set otherwise; this happens to $s_n$ on the $T$-failed ending stages
and to $s_t$ on the $T$-successful ending stages.
\end{list}

\begin{table}
\colorbox{black!5}{\arrayrulecolor{green!50!black} 
  \begin{tabular}{rl}
\cmidrule[0.5pt]{1-2}
{\small $w_{\sigma}$}  & {\small wager of $M$ at $\sigma$} \\[0.5ex]
{\small $v_{\sigma}$}  & {\small wager of $T$ at $\sigma$} \\[0.5ex]
{\small $n_{\sigma}$}  & {\small wager of active $N$ at $\sigma$} \\[0.5ex]
{\small $r_{\sigma}$}  & {\small $T(\sigma)-M(\sigma)$} \\[0.5ex]
\cmidrule[0.5pt]{1-2}
\end{tabular}}\centering
\hspace{0.2cm}
\colorbox{black!5}{\arrayrulecolor{green!50!black} 
\begin{tabular}{rl}
\cmidrule[0.5pt]{1-2}
{\small $c_{\sigma}$}  & {\small marker of $r$ at the starting stage of a sub-cycle} \\[0.5ex]
{\small $i_{\sigma}$}  & {\small index $\rho$ of backup strategy $N_{\rho}$ active at $\sigma$} \\[0.5ex]
{\small $s_t,s_n$}  & {\small saving functions for $T,N$ respectively} \\[0.5ex]
{\small $h$}  & {\small function with $w_{\sigma}\leq h(|\sigma|)$, $\lim_n h(n)=0$} \\[0.5ex]
\cmidrule[0.5pt]{1-2}
\end{tabular}}\centering
\caption{Parameters for the proof of Theorem \ref{lL3i4uiQDZ} in \S\ref{iPgHQCWM76}.}\label{iHCIFQhir}
\vspace{-0.5cm}
\end{table}

We now inductively define the wagers $v_{\sigma},n_{\sigma}$ of the strategies $T,N$ along 
with their saving functions $s_t,s_n$ and 
the type of the stage $\sigma$.\footnote{The subscript in a wager denotes the argument for the wager function --
not any suppressed index of the corresponding strategy.}

{\bf Construction.}
Let $T(\lambda)=M(\lambda)+1$, then $r_{\lambda}=1$.
Let $c_{\lambda}=r_{\lambda}$, $i_{\lambda}=\lambda$, and $\lambda$ be neutral.

Given $\sigma\neq\lambda$, 
inductively assume that $c_{\hat{\sigma}}, i_{\hat{\sigma}}$ and the type of $\hat{\sigma}$
have been defined, consider the cases:
\begin{enumerate}[\hspace{0.1cm}(a)]
\item {\bf if $\hat{\sigma}$ is neutral stage:} Let 
$v_{\hat{\sigma}}=n_{\hat{\sigma}}=w_{\hat{\sigma}}$,
$c_{\sigma}=c_{\hat{\sigma}}$, $i_{\sigma}=i_{\hat{\sigma}}$.
And
\begin{itemize}
\item if $ N (\sigma) > 2 \ceil{2 / c_{\sigma} } $ and $h(|\sigma|) < \min \{1, c_{\sigma}/ 4\}$,
let $\sigma$ start a $i_{\sigma}$-cycle;
\item otherwise, let $\sigma$ be a neutral stage.
\end{itemize}
\item {\bf if $\hat{\sigma}$ is a (sub-)cycle stage:} Let 
$v_{\hat{\sigma}}=2w_{\hat{\sigma}}$, $n_{\hat{\sigma}}= - \ceil{2 / c_{\hat{\sigma}} } \cdot w_{\hat{\sigma}}$ and:
\begin{itemize}
\item if $r^0_{\sigma}\in (c_{\hat{\sigma}} / 2, c_{\hat{\sigma}} + 1)$, let $c_{\sigma}=c_{\hat{\sigma}}$,
$i_{\sigma}=i_{\hat{\sigma}}$ and $\sigma$ be a sub-cycle stage;
\item if $r^0_{\sigma}\leq c_{\hat{\sigma}} / 2 $, mark $\hat{\sigma}$ as a {\em $T$-failed ending stage}
and let $s_{n}(\hat{\sigma})=1$, $c_{\sigma}= r_{\sigma} $,
$i_{\sigma} = i_{\hat{\sigma}}$, and $\sigma$ be a neutral stage;
\item if $r^0_{\sigma}\geq c_{\hat{\sigma}} + 1$, mark $\hat{\sigma}$ as a {\em $T$-successful ending stage}
and let $s_{t}(\hat{\sigma})=1$, $c_{\sigma}= r_{\sigma}$,
$i_{\sigma} = \sigma$, and $\sigma$ be a neutral stage.
\end{itemize}
\end{enumerate}

\subsection{Verification of hedging strategy for bold play: concluding the proof of Theorem \ref{lL3i4uiQDZ}}
Recall our convention \eqref{LVSHDGcf4N} regarding the subscripts. 
A direct consequence of the construction in \S\ref{CHwOEWwFd} is:
\begin{equation}\label{lem:neutpro}
\parbox{14cm}{if $\sigma \neq \lambda$ and $\hat{\sigma}$ is a neutral stage then 
$N(\sigma) -  M(\sigma) =  N(\hat{\sigma}) - M(\hat{\sigma}) $ and 
	$r_{\sigma} = r_{\hat{\sigma}} =c_{\hat{\sigma}} = c_{\sigma}$.}
\end{equation}
{\em Proof of \eqref{lem:neutpro}.}
The first equality follows directly by the assignment
$v_{\hat{\sigma}}=n_{\hat{\sigma}}=w_{\hat{\sigma}}$ which takes place in clause (a) which is
dedicated to neutral stages.
The same clause sets
 $c_{\sigma}=c_{\hat{\sigma}}$. 
 For the remaining equalities recall that
 $r_{\sigma}$ is simply a notation for $T(\sigma)-M(\sigma)$.
 Since $v_{\hat{\sigma}}=w_{\hat{\sigma}}$ and no saving occurs on a neutral stage 
 $\hat{\sigma}$ we have
 $r_{\sigma} = r_{\hat{\sigma}}$. Finally $c_{\hat{\sigma}}=r_{\hat{\sigma}}$ follows from the fact
 that the only way that (a) can be accessed in the construction loop is from one of the
 last two clauses of (b),
both of which set $c_{\sigma}=r_{\sigma}$ (and during this transition $\sigma$ becomes $\hat{\sigma}$).

\begin{lem}[Parameters at ending stages]\label{lem:cyclecha}
	For any $\sigma \neq \lambda$ such that $\hat{\sigma}$ is a sub-cycle stage
	and $\eta $ is the starting stage of that sub-cycle interval, 
\begin{enumerate}[\hspace{0.5cm}(i)]
	\item if $\hat{\sigma}$ is not an ending stage, $r_{\sigma} > r_{\eta}/2$ and $N(\sigma) > \ceil{2 / r_{\eta} } $;
	\item if $\hat{\sigma}$ is a $T$-failed ending stage, $ r_{\sigma} > r_{\eta}/4 $, $N$ saves 1 and $N(\sigma) - M(\sigma) \ge  N(\eta) - M(\eta) $;
	\item if $\hat{\sigma}$ is a $T$-successful ending stage, $r_{\sigma} \ge r_{\eta}$, $N(\sigma) > 0$ and $T$ saves 1.
\end{enumerate}
\end{lem}

\begin{proof}
First we observe that $c_{\tau}$ does not change in a sub-cycle stage $\tau$.
Given the hypothesis of the lemma about $\sigma,\eta$ and the construction,
we have $T(\sigma) - T(\eta) =  2 \cdot (M(\sigma) - M(\eta))$ so
\begin{eqnarray*}
r^0_{\sigma} - r_{\eta} = & \big(T(\sigma) - M(\sigma)\big) - \big(T(\eta) - M(\eta)\big) &= M(\sigma) - M(\eta)\\[0.2cm]
N^0(\sigma) - N(\eta) = & - \ceil{2 / c_{\eta} } \cdot 
\big(M(\sigma) - M(\eta)\big) &=  - \ceil{2 / c_{\eta} } \cdot (r^0_{\sigma} - r_{\eta}).
\end{eqnarray*}
On the other hand, by \eqref{lem:neutpro}, $c_{\eta} = r_{\eta}$ and by the construction, 
$N (\eta) > 2 \ceil{2 / r_{\eta} }$, $h(|\eta|) < \min \{1, r_{\eta}/ 4\}$.
 
Given these facts we may proceed to the proof of the clauses of the lemma, starting with (i).
If $\hat{\sigma}$ is not an ending stage, by the construction we have
$r_{\sigma}\in (c_{\hat{\sigma}} / 2, c_{\hat{\sigma}} + 1) = (c_{\eta} / 2, c_{\eta} + 1) = (r_{\eta} / 2, r_{\eta} + 1) $, so
\[
N(\sigma) - N(\eta) =  - \ceil{2 / c_{\eta} } \cdot (r_{\sigma} - r_{\eta}) > - \ceil{2 / c_{\eta} }
\hspace{0.5cm}\textrm{so} \hspace{0.5cm}
N(\sigma) > N(\eta) - \ceil{2 / c_{\eta} } > \ceil{2 / r_{\eta} }
\]
which concludes the proof  of (i). For the last two clauses, assuming that
 $\hat{\sigma}$ is an ending stage, we have
$|r^0_{\sigma} - r_{\hat{\sigma}}| = |M(\sigma) - M(\hat{\sigma})| = |w_{\hat{\sigma}}| < h(|\hat{\sigma}|) \le h(|\eta|)$
and $r_{\hat{\sigma}} \in (r_{\eta} / 2, r_{\eta} + 1)$, so
\[
r^0_{\sigma} \in \Big(r_{\eta} / 2 - h(|\eta|) , r_{\eta} + 1 + h(|\eta|)\Big) \subseteq \big(r_{\eta} / 4 , r_{\eta} + 2\big).
\]
If $\hat{\sigma}$ is a $T$-failed ending stage then $r^0_{\sigma}\leq c_{\hat{\sigma}} / 2 = r_{\eta} / 2$ and $N$ saves 1.
Moreover $r_{\sigma} = r^0_{\sigma} \in (r_{\eta} / 4 , r_{\eta} / 2]$
and 
\[
N(\sigma) - N(\eta) =  - \ceil{2 / c_{\eta} } \cdot (r_{\sigma} - r_{\eta}) - 1 
\ge \ceil{2 / c_{\eta} } \cdot r_{\eta} / 2 - 1 \ge 0 \ge M(\sigma) - N(\eta)
\]
which shows that $N(\sigma) - M(\sigma) \ge  N(\eta) - M(\eta)$ as required for (ii).
Finally if $\hat{\sigma}$ is a $T$-successful ending stage, 
we have $r^0_{\sigma}\geq c_{\hat{\sigma}} + 1 = r_{\eta} + 1$ and $T$ saves 1.
Moreover  $r_{\sigma} = r^0_{\sigma} -1 \in [r_{\eta}  , r_{\eta} + 1)$
and 
\[
N(\sigma) - N(\eta) =  - \ceil{2 / c_{\eta} } \cdot (r^0_{\sigma} - r_{\eta})
> - 2 \ceil{2 / c_{\eta} }
\]
which shows that $N(\sigma) >  N(\eta) - 2 \ceil{2 / c_{\eta} } > 0$ as required for (iii).
\end{proof}

\begin{lem}
Strategies $T$ and  $N_{\rho}, \rho \in I$ are $g$-granular supermartingales.
\end{lem}
\begin{proof}
As the wagers of $T$ and $N_{\rho}$ are always integer multiples of the wagers of $M$,
the $g$-granularity of $T, N_{\rho}$ follows from the $g$-granularity of $M$.
Then by the construction, we only need to verify that $T$ and $N_{\rho}$ always have non-negative value.
As $r_{\lambda} = 1 > 0$, then by \eqref{lem:neutpro} and Lemma~\ref{lem:cyclecha} inductively
we easily get $r_{\sigma} > 0 $ for all $\sigma$.
Then $T(\sigma) = M(\sigma) + r_{\sigma} > 0 $ for all $\sigma$, as required.
Fix $\rho \in I$, for simplicity we drop the subscript $\rho$ for $N_{\rho}$ for the rest of this proof.
First for all $\sigma \nsucceq \rho$, $N(\sigma) = N(\rho) = M(\rho) \ge 0$.
And by Lemma~\ref{lem:cyclecha} $N(\sigma) > 0$ for all $\sigma$ such that 
$\hat{\sigma}$ is a sub-cycle stage but not $T$-failed ending stage.
Moreover, if $\hat{\sigma}$ is a $T$-successful ending stage, $N$ has positive value,
then for all the ending stages of $N$ it also has positive value.
On the other hand, as $\rho$ is a neutral stage for $N$ and $N(\rho) - M(\rho) = 0$,
by Lemma~\ref{lem:neutpro} and Lemma~\ref{lem:cyclecha} inductively we get that for all $\sigma$
such that $\hat{\sigma}$ is a neutral stage or $T$-failed ending stage, $N(\sigma) - M(\sigma) \ge 0$,
i.e., $N(\sigma) \ge M(\sigma) \ge 0$.
\end{proof}

\begin{lem}[Sub-cycles]\label{czfV7tMnqi}
Along any path $X$ of outcomes such that
$\limsup_n M(X\restr_n)=\infty$, there are infinitely many sub-cycles starting and ending along $X$.
\end{lem}
\begin{proof}
If only finitely many sub-cycles occur along $X$, one of the following must hold:
\begin{enumerate}[\hspace{0.5cm}(i)]
\item almost all prefixes of $X$ are neutral;
\item there exists a sub-cycle along $X$ which never ends. 
\end{enumerate}
It remains to show each of the above clauses implies  $\limsup_n M(X\restr_n)<\infty$.
First assume that (i) holds and that $\eta$ is the least prefix of $X$ such that 
all prefixes of $X$ after $\eta$ are neutral. 
If $\rho$ is the index of $\eta$, then for all $\eta \preceq \sigma \prec X$
we have $i_{\sigma}=\rho$ and $c_{\sigma}=c_{\eta}$. 
As $h \rightarrow 0$, then there is $\eta \preceq \tau \prec X$ such that 
$h(|\tau|) < \min \{1, c_{\eta}/ 4\} = \min \{1, c_{\tau}/ 4\}$.
Moreover, then for all $\tau \preceq \sigma \prec X$, 
$h(|\sigma|) < \min \{1, c_{\sigma}/ 4\}$.
As no sub-cycle of $N_{\rho}$ starts at any prefix of $X$ after $\eta$,
then for all $\tau \preceq \sigma \prec X$, 
$ N (\sigma) \le 2 \ceil{2 / c_{\sigma} } = 2 \ceil{2 / c_{\eta} } $.
Hence $M$ is bounded above along $X$, as required.
Second, assume that (ii) holds and at $\rho\prec X$ a sub-cycle starts, which never ends along $X$.
Then for all $\eta\prec X$ after $\rho$ we have $r_{\eta}-r_{\rho}=M(\eta)-M(\rho)$
By condition for ending a sub-cycle in the construction,
it follows that $r_{\eta}$ remains bounded above by $c_{\rho} + 1 $ along $X$, hence
$M$ is bounded above along $X$, as required.
\end{proof}

\begin{lem}
Along any path $X$ of outcomes such that
$\limsup_n M(X\restr_n)=\infty$, one of the following holds:
\begin{enumerate}[\hspace{0.5cm}(a)]
\item there are infinitely many $T$-successful sub-cycles along $X$;
\item all but finitely many sub-cycles along $X$ are $T$-failed.
\end{enumerate}
If (a) holds, then $T$ saves successfully along $X$. If (b) holds, then there exists $\rho\prec X$
such that $N_{\rho}$ saves successfully along $X$, where it remains active.
\end{lem}
\begin{proof}
By the assumption about $X$ and Lemma \ref{czfV7tMnqi} there are infinitely many
sub-cycles along $X$. Since a new sub-cycle only starts after the previous one has ended,
and since each sub-cycle ends either $T$-successfully or in $T$-failure, it follows that
either (a) or (b) holds along $X$.
If (a) holds, then by Lemma~\ref{lem:cyclecha}, $T$ saves successfully along $X$.
If (b) holds, the indices of the initial segments of $X$ reach a limit $\rho$.
Hence starting from $\rho$ and along $X$, there will be infinitely many sub-cycles of $N_{\rho}$
and all of them will end in $T$-failure. Hence starting from $\rho$ and along $X$, 
strategy $N_{\rho}$ will remain active and
by Lemma~\ref{lem:cyclecha} it will successfully save along $X$.
\end{proof}

\section{Conclusion and brief discussion}\label{XXYGMo8eSb}
Liquidity  in betting situations, in the sense of infinite divisibility of the capital, 
allows for certain flexibilities in the strategies, including avoiding bankruptcy while placing
infinitely many bets, and saving an unbounded capital on the condition that the betting strategy is
successful. Such properties are based on the fact that liquidity allows arbitrary {\em scaling} of the
strategy, \ie the implementation of essentially the original strategy but with arbitrarily small available
capital.
A recent line of research 
\cite{cie/BienvenuST10,TeutChalcraft,Peretzwager, Peretzagainst,Teu14agCWGnp}
studied betting strategy without this property, where the wagers are restricted in certain ways.
For example,  \cite{Teu14agCWGnp} showed that successful saving
is not always possible in the presence of a fixed minimum wager restriction.
In the present work we studied the the possibility and impossibility of saving
in the presence of inflation or equivalently, as discussed in \S\ref{1JyMSDjB2p},
the presence of a shrinking minimum wager restriction. We found that there is a dichotomy
in the properties of such strategies, which is defined in terms of the rate of decrease of the
minimum wager $2^{-g(n)}$: {\em fine granularity} where $\sum_n 2^{-g(n)}$ is finite, and
{\em coarse granularity} where the sum is infinite.
In general, fine granularity allows for saving (subject to successful betting) while
coarse granularity does not. On the other hand, in the latter impossibility case, we found that
by employing aggressive (bold) saving strategies that have in total access
to unbounded capital, it is possible to ensure successful saving on any possible outcome sequence
where a given timid betting strategy succeeds.

Given the role of bold versus timid betting in our results, we would like to pose the question
whether this qualification is necessary in a complete analysis of the possibility of saving in
strategies with restricted wagers, or there is a classification that does not depend on it.
Another direction for exploration would be
the establishment of explicit connections of this recent line of research, with the two
more classic approaches that we discussed in \S\ref{1JyMSDjB2p}, based on
integration and game-theoretic probability. Although we could not identify
direct links in the different mathematical frameworks for betting, there is a considerable intersection
on the main themes that are studied. 

Our results can also be viewed in
the context of algorithmic randomness.
One of the standard approaches to the formalization of algorithmic randomness 
of infinite binary sequences is
based on betting strategies and was pioneered in \cite{Schnorr:71,Schnorr:75}. The intuitive idea
here is that algorithmically random sequences should be sequences of binary outcomes
on which no `effective' betting strategy can succeed.
This approach is essentially
equivalent to earlier formalizations in terms of statistical tests in \cite{MR0223179}
or compression in \cite{MR0366096}. In general, for each choice of a countable collection of
strategies as the {\em effective betting strategies}, we get a corresponding randomness notion.
In this way, various restrictions on the notion of betting and success of strategies
correspond to different strength of algorithmic randomness
(\eg see \cite[\S 6, \S 7]{rodenisbook} or \cite[\S 7]{Ottobook}).
Similarly, by
interpreting granularity as a feasibility condition on the strategies,
we obtain notions of {\em randomness against granular strategies}.
Considering {\em saving strategies} as opposed to {\em betting strategies},
we may view some of our results as separations or equivalences of the 
corresponding randomness notions.

\bibliographystyle{abbrvnat}
\bibliography{granuma}
\end{document}